\ifdefined\XeTeXversion\else\pdfoutput=1\fi
\documentclass[11pt]{article}

\usepackage[utf8]{inputenc}
\usepackage[T1]{fontenc}
\usepackage{lmodern}
\usepackage[margin=1in]{geometry}
\usepackage{amsmath,amssymb,amsthm}
\usepackage{booktabs}
\usepackage{enumitem}
\usepackage{listings}
\usepackage{xcolor}
\usepackage[colorlinks=true,linkcolor=blue!50!black,citecolor=green!40!black,urlcolor=blue!60!black]{hyperref}

\newcommand{\poly}{\operatorname{poly}}

\newcommand{\doi}[1]{DOI: \href{https://doi.org/#1}{\nolinkurl{#1}}}

\theoremstyle{plain}
\newtheorem{theorem}{Theorem}[section]
\newtheorem{lemma}[theorem]{Lemma}
\newtheorem{corollary}[theorem]{Corollary}
\newtheorem{proposition}[theorem]{Proposition}
\newtheorem{observation}[theorem]{Observation}
\newtheorem*{projthm}{Projection Theorem}
\theoremstyle{remark}

\theoremstyle{definition}
\newtheorem{example}{Example}[section]
\newtheorem*{syncprop}{Boundary synchronization (SI)}

\title{Exact Algorithms and Complexity of the Smallest Grammar Problem}
\author{Chen Xiangle\\[2pt]\small Tsinghua University\\\small\texttt{chenxiangle25@mails.tsinghua.edu.cn}}
\date{27 October 2026}

\begin{document}
\maketitle

\begin{abstract}
The Smallest Grammar Problem (SGP) asks for a smallest straight-line grammar generating a given string, where the size of a grammar is the total length of its right-hand sides. Our algorithms rest on a structural fact: every string has an optimal grammar in which one can choose an occurrence of the expansion of each auxiliary nonterminal so that the chosen occurrences are pairwise disjoint. Hence the dictionary of an optimal grammar can be read off from a partition of the input, and once the dictionary is known, the rules can be computed by dynamic programming.

Our main result is a deterministic exact algorithm for strings of length $n\ge2$ over at most $\sigma$ symbols, where $2\le\sigma\le n$. Writing $R=\log_\sigma n$, it runs in polynomial space and time
\[
 n^{O(1)}2^{O\left(\frac{n\log\log n}{R\sqrt{\log n}}\right)}
 =n^{O(1)}2^{O\left(\frac{n\log\sigma\,\log\log n}{(\log n)^{3/2}}\right)}.
\]
For fixed alphabets this is $2^{O(n\log\log n/(\log n)^{3/2})}$, and for arbitrary alphabets $2^{O(n\log\log n/\sqrt{\log n})}$, up to polynomial factors. The algorithm encodes disjoint-occurrence certificates as a sparse tensor network and applies a Boolean tensor contraction theorem from a recent OpenAI preprint.

Without this external theorem, the same structural fact yields self-contained linear-space algorithms running in $O^*(1.872258^n)$ time, improving the $O^*(3^n)$ bound of Casel et al., and in $2^{O((n/R)\log(R+1))}$ time. On the complexity side, we prove that SGP is NP-complete over a binary alphabet, answering a question of Casel et al., and give reductions from Vertex Cover with tunable alphabet size. These reductions turn any faster SGP algorithm into an explicit SAT algorithm; we state these transfers unconditionally, including the loss from sparsifying general formulas.
\end{abstract}


\section{Introduction}

When a string contains many repeated fragments, a straight-line grammar can define each fragment once and refer to it thereafter through a nonterminal. Since rules may refer to other rules, repetitions nested several levels deep can be represented as well. Such grammars are the basic representation in grammar-based compression and also support indexing directly on the compressed text~\cite{CNP21}. A smallest grammar for a string $T$ marks the limit of this form of compression and serves as the benchmark for approximation algorithms. Finding it is the \emph{Smallest Grammar Problem} (SGP); throughout, the size of a grammar is the total number of symbols on the right-hand sides of its rules.

It is precisely the nesting of rules that makes exact algorithms difficult. Two useful fragments may contain one another, their occurrences in the input may overlap, and once a fragment is selected, one must still decide which shorter fragments define it. A direct enumeration of rules must therefore handle two issues at once: which dictionary to use, and how the rules depend on one another. Charikar et al.~\cite{CLL05} studied approximation algorithms and the complexity of SGP. Casel et al.~\cite{CFGGS21} showed that the problem remains NP-complete even over fixed alphabets of $24$ or $17$ symbols, and gave an $O^*(3^n)$-time exact algorithm for general grammars, where $n=|T|$.

\subsection{The structural lemma}

All our algorithms are based on the following property of optimal grammars. The \emph{dictionary} of a grammar is the set of strings generated by its auxiliary nonterminals, that is, by all nonterminals other than the start symbol.

\begin{lemma}[Disjoint representatives; Theorem~\ref{thm:disjoint}]\label{lem:intro-disjoint}
Every string $T$ has an optimal grammar $G$ with the following property: for each auxiliary nonterminal $A$ of $G$ one can choose an occurrence of its expansion in $T$ such that the chosen occurrence intervals are pairwise disjoint.
\end{lemma}

The proof is short. Traverse the derivation tree and expand each nonterminal only at its first encounter. In an optimal grammar every auxiliary nonterminal is referenced at least twice, so it also appears as an unexpanded leaf, and the leaves of the pruned tree occupy disjoint intervals of $T$. The lemma decouples the two issues above, and it has three consequences that we use repeatedly.
\begin{enumerate}[label=(\roman*),leftmargin=*]
\item \emph{The dictionary can be read off from a partition of $T$.} Take the representatives as blocks and every other position as a block of its own. For a fixed dictionary, a shortest parsing of every rule can be computed by dynamic programming (Lemma~\ref{lem:parse}), so enumerating all $2^{n-1}$ partitions already solves SGP in $O^*(2^n)$ time (Theorem~\ref{thm:2n}).
\item \emph{The dictionary has total length at most $n$.} In particular, there are at most $n/L$ dictionary strings of length at least $L$, while over an alphabet of size $\sigma$ there are fewer than $\sigma^L$ distinct strings shorter than $L$.
\item \emph{An optimal grammar has a certificate made of input intervals.} Placing every rule on its representative, each rule definition and each reference becomes an interval of $T$. A grammar of size $K$ thus selects only $O(K)$ intervals among polynomially many candidates.
\end{enumerate}
Consequence~(i) leads to the $O^*(1.872258^n)$ algorithm, (ii) to the alphabet-sensitive enumeration, and (iii) to our main algorithm.

\subsection{Our results}

Let $g(T)$ denote the size of a smallest grammar for $T$, and let $n$ be the input length. All our algorithms handle fully nested rules and output an optimal grammar.

\paragraph{The main algorithm.}
Our main result uses consequence~(iii) together with the exact-counting and tensor-contraction results of a recent OpenAI preprint~\cite{sat}. We encode the certificate conditions as local tensors, index only the sparse set of selected wires, and apply the contraction theorem with explicit bounds on degree and arithmetic precision. A simple construction bounds the optimum by $K=O(n/\log_\sigma n)$ (Lemma~\ref{lem:sparse-size-bound}), and this budget gives an algorithm that is uniform over the full alphabet range.

\begin{theorem}[Uniform alphabet-sensitive algorithm]\label{thm:subexp-intro}
Let $n\ge2$, $2\le\sigma\le n$, and $R=\log_\sigma n$.
Using the Boolean tensor contraction theorem of OpenAI~\cite{sat}, SGP on
strings over at most $\sigma$ symbols can be solved deterministically in
polynomial space and time
\[
 n^{O(1)}2^{O\left(
 \frac{n\log\log(n+2)}{R\sqrt{\log(n+2)}}
 \right)}.
\]
All constants are independent of $\sigma$. Unary inputs are covered by
taking $\sigma=2$.
\end{theorem}

The exponent is $o(n/R)$ uniformly across the alphabet range. The proof in
Section~\ref{sec:sparse-interval-tensor} gives a more detailed bound in terms of a size budget $K$ and a
length threshold $L$,
\[
 n^{O(1)}2^{O\left(
 \frac{n\log(L+1)}L+
 \frac{K\log(\mathrm e nL/K)}{\sqrt{\log n}}
 \right)},
\]
where $g(T)\le K\le n$, $K\ge n/\log_2(n+2)$, and
$2\le L\le(\log_2(n+2))^2$ with $L<n$.
The uniform result follows with $K=\Theta(n/R)$ and
$L=\Theta(R\sqrt{\log n})$. Table~\ref{tab:regimes} lists the resulting exponents for several alphabet sizes.

\begin{table}[ht]
\centering\small
\begin{tabular}{@{}lll@{}}
\toprule
Alphabet bound & Theorem~\ref{thm:subexp-intro} & Theorem~\ref{thm:main}, part~2 \\
\midrule
fixed $\sigma\ge2$ & $nu/t^{3/2}$ & $nu/t$ \\
$\sigma=\Theta(t^a)$ & $nu^2/t^{3/2}$ & $nu^2/t$ \\
$\log\sigma=\Theta(t^\beta)$ & $nu/t^{3/2-\beta}$ & $nu/t^{1-\beta}$ \\
$\sigma=\Theta(n^\alpha)$ & $nu/\sqrt t$ & $n$ \\
\bottomrule
\end{tabular}
\caption{Times have the form $n^{O(1)}2^{O(E)}$, with the indicated exponent
$E$, $t=\log n$, and $u=\log\log n$. The parameters $a>0$, $0<\beta<1$,
and $0<\alpha\le1$ are fixed. The first column uses polynomial space, the second linear space.
The bounds apply to strings using at most the indicated number of symbols.}
\label{tab:regimes}
\end{table}

Smaller alphabets give smaller exponents, and even the largest alphabet
bound gives subexponential time. These are worst-case guarantees; for
example, a string whose $n$ symbols are all distinct is easy to recognize,
and the literal grammar is optimal for it.

Theorem~\ref{thm:subexp-intro} depends on the external contraction theorem, which appears in a preprint that has not been peer reviewed. Section~\ref{sec:sparse-interval-tensor} states all of its hypotheses that we use; we do not reprove it. Exact computation of grammar-related compression measures via SAT solvers was studied earlier~\cite{bannai}; our algorithm instead encodes interval certificates directly as a tensor network and exploits their sparsity.

\paragraph{Self-contained linear-space algorithms.}
The remaining algorithmic results do not use the external theorem. Their running times are larger than in Theorem~\ref{thm:subexp-intro}, but their proofs are self-contained and they use only linear space.

\begin{theorem}[Linear-space algorithms]\label{thm:main}\leavevmode
\begin{enumerate}[label=\arabic*.,leftmargin=*]
\item SGP can be solved exactly in $O^*(\rho^n)$ time and $O(n)$ machine words of space, and an optimal grammar can be output within the same asymptotic bounds, where
\[
\rho=1.8722572196\ldots<2.
\]
\item Let $n\ge2$, $2\le\sigma\le n$, and
\[
R=\log_\sigma n=\frac{\log n}{\log\sigma}.
\]
For every input of length $n$ over at most $\sigma$ symbols, SGP can be solved exactly in
\[
2^{O\!\left(\frac{n}{R}\log(R+1)\right)}
\]
time and $O(n)$ machine words of space.
\end{enumerate}
\end{theorem}

Part~1 improves the $O^*(3^n)$ bound of~\cite{CFGGS21} for general grammars to $O^*(1.872258^n)$. The smaller base $1.8392$ in~\cite{CFGGS21} applies only to $1$-level grammars, the restricted case in which only the start rule may refer to nonterminals; our algorithm handles fully general nested grammars.

For part~2, it is helpful to view $R$ as a number of digits: writing $n$ in base $\sigma$ takes exactly $\lfloor R\rfloor+1$ digits, which differs from $R$ by at most one. This interpretation matches the counting in the algorithm: there are $\sigma^L$ strings of length $L$ and $\sigma^R=n$, and the algorithm compares $\sigma^L$ with $n$ to decide which strings to enumerate by content and which by position.

\paragraph{Complexity.}
Small alphabets help the algorithms above, but they do not make the problem easy: SGP remains NP-complete even over two symbols.

\begin{theorem}\label{cor:binary}
The decision version of SGP over the binary alphabet $\{0,1\}$ is NP-complete.
\end{theorem}

This answers the question on binary alphabets posed by Casel et al.~\cite[Section~6.1]{CFGGS21}: computing a smallest grammar is NP-hard even for inputs over two symbols. Our proof starts from the $24$-letter reduction of~\cite{CFGGS21}, establishes a constant additive bound between the optima before and after encoding, and then takes a constant number of copies of the source graph to separate the decision threshold from this constant error. Masierak's preprint~\cite{Mas26} also reports NP-completeness over a binary alphabet under the standard measure of total right-hand side length, via a different route that starts from the assembly index and then converts between size measures.

To compare algorithms at different alphabet sizes, we also construct reductions whose alphabet size can be chosen freely.

\begin{theorem}[Reductions with tunable alphabet size]\label{thm:tunable}
For every polynomial-time computable integer function $b(N)\in[2,2N]$, there is a polynomial-time reduction from Vertex Cover on graphs of maximum degree at most three to SGP. For a source graph with $N$ vertices and at least one edge, the target string satisfies
\[
\sigma=|\operatorname{alph}(T)|=2b(N)+26,\qquad
n=|T|=\Theta\!\bigl(N(1+\log_{b(N)}N)\bigr).
\]
The constants in the length bound are independent of $b$. Moreover,
$g(T)=379N+3M+5+\tau(G)$, where $M$ is the edge count. Thus $n/R=\Theta(N)$ uniformly in $b$.
\end{theorem}

The quantity $n/R$ is an upper bound on the optimum grammar size, and in
this reduction it equals the size of the source graph up to absolute
constants. It is not a lower bound on the running time: the exponent in
Theorem~\ref{thm:subexp-intro} is $o(n/R)$.

We do not phrase these reductions as lower bounds under the Exponential Time Hypothesis (ETH)~\cite{IPZ01}, since the preprint~\cite{sat} claims a refutation of ETH. Their quantitative content is unconditional: Section~\ref{sec:transfer} states how a faster SGP algorithm would yield a faster SAT algorithm, with explicit length and alphabet parameters. For example, a binary SGP algorithm with exponent $O(n(\log\log n)^a/(\log n)^{1+\delta})$ would give an algorithm for fixed-density $3$-SAT with exponent $O(N(\log\log N)^a/(\log N)^\delta)$ (Corollary~\ref{cor:log-speed-scale}). These statements hold whether or not ETH is true.

\subsection{Techniques}

\paragraph{Linear-space algorithms.}
Section~\ref{sec:disjoint} proves Lemma~\ref{lem:intro-disjoint}; the pruning used there is exactly the grammar tree of~\cite{CNP21}. Section~\ref{sec:longshort} goes one step further and selects, among all optimal grammars, one with the fewest rules; in such a grammar every dictionary string of length two is referenced at least three times. This additional condition substantially reduces the number of candidates for short dictionary strings and is what brings the base of the exponent below two.

Section~\ref{sec:alphalg} splits the dictionary by length and enumerates the two parts differently. Short strings are enumerated by content, and their number is bounded in terms of the alphabet size; long strings are enumerated by the positions of their disjoint representatives, and their number is bounded in terms of the input length. For sufficiently large $R$, with a length threshold of about $R/2$, there are $O(\sqrt n)$ short strings and $O(n/R)$ long representatives. The logarithmic overhead of enumerating these representatives accounts for the factor $\log(R+1)$ in part~2 of Theorem~\ref{thm:main}.

\paragraph{The tensor network.}
Section~\ref{sec:sparse-interval-tensor} uses three families of intervals:
dictionary representatives, tokens of the start rule, and tokens of the
auxiliary rules. Intervals longer than a threshold $L$ are enumerated
explicitly, in $2^{O(n\log(L+1)/L)}$ time. Shorter intervals become Boolean
choices in a network of position and content tensors. The position tensors enforce
coverage and rule boundaries, and the content tensors match references to
their definitions. Contents with many candidate occurrences are split into blocks with
small summary tensors, which keeps the degree low.

A certificate of cost at most $K$ sets at most $8K$ wires to one,
including the summary wires. Grouping the tensors and indexing only the
subsets of wires set to one reduces the network to
$O(K\log(\mathrm e nL/K))$ Boolean wires. Truncating the ring in a cost
variable enforces the budget, and an explicitly bounded modulus makes the
count exact, so that a positive count is never mistaken for zero. With these
bounds the tensor-contraction theorem of~\cite{sat} applies directly, and
self-reduction recovers an optimal grammar.

\paragraph{Encodings.}
The complexity results use a different tool. The main technical difficulty is translating a grammar for an encoded string back into a grammar for the source string, since the encoded grammar may reuse fragments of codewords. Section~\ref{sec:anchors} places in each codeword an anchor labeled with the corresponding source symbol and proves a key property: if some occurrence of a fragment contains at least two anchors, then every occurrence of the fragment yields the same sequence of labels. This property allows each target nonterminal to be translated consistently into a rule of the source grammar. For fixed alphabets, the encoding overhead is bounded by a constant, so Section~\ref{sec:npc} needs only a constant number of copies of the source graph to separate this error from the unit gaps caused by the vertex cover size. Growing alphabets are harder because this constant bound is no longer available. Section~\ref{sec:eth} therefore prepends a carefully designed prefix in which every codeword occurs repeatedly. This forces an optimal grammar to create rules for all codewords, so that the encoding overhead can be computed exactly rather than merely bounded.

\subsection{Organization}

Section~\ref{sec:prelim} introduces reduced grammars and computes optimal parsings for a fixed dictionary. Section~\ref{sec:disjoint} proves the structural lemma and gives the $O^*(2^n)$ algorithm, and Section~\ref{sec:longshort} proves Theorem~\ref{thm:main}. Section~\ref{sec:sparse-interval-tensor} proves Theorem~\ref{thm:subexp-intro}. Sections~\ref{sec:alphabet} and~\ref{sec:eth} prove Theorems~\ref{cor:binary} and~\ref{thm:tunable}, and Section~\ref{sec:transfer} derives the transfers to SAT. Section~\ref{sec:discussion} discusses open problems. The appendices bound reduction lengths, give an explicit binary code, and report finite computer checks of the tensor encoding.

\section{Preliminaries}\label{sec:prelim}

\subsection{Strings and grammars}

Let $T=T[1]\cdots T[n]\in\Sigma^n$ with $n\ge1$. For an interval $I=[i,j]\subseteq[1,n]$ we write $T[I]=T[i]\cdots T[j]$ and $|I|=j-i+1$. Two intervals are \emph{disjoint} if they share no position. A contiguous substring is also called a \emph{factor}, and a \emph{proper} substring is one strictly shorter than the whole string. Throughout, a dictionary entry is a string, whereas an occurrence of it in $T$ is an interval; a dictionary entry may have many occurrences, and this distinction will be important.

A straight-line grammar $G$ has a set $V$ of nonterminals and a start symbol $S\in V$; every $A\in V$ has exactly one rule $A\to\alpha_A$ with $\alpha_A\in(V\cup\Sigma)^+$. The dependency relation is acyclic, so every nonterminal expands to a unique string $\operatorname{val}(A)$. We require $\operatorname{val}(S)=T$ and define
\[
|G|=\sum_{A\in V}|\alpha_A|,\qquad
g(T)=\min_{\operatorname{val}(S)=T}|G|.
\]
We distinguish two lengths: $|\alpha_A|$ is the number of symbols on a right-hand side, and $|\operatorname{val}(A)|$ is the number of characters in the expansion. Nonterminals other than $S$ are called \emph{auxiliary}, and their expansions form the \emph{dictionary}
\[
\mathcal D(G)=\{\operatorname{val}(A):A\in V\setminus\{S\}\}.
\]

$O^*(\cdot)$ hides factors polynomial in $n$. Space is measured in machine words of $\Theta(\log n)$ bits. Symbols of a general alphabet can first be renamed to integers in $[1,n]$. Unless stated otherwise, $\log$ is to base two. For a graph $G$, $\tau(G)$ denotes the size of a minimum vertex cover.

\subsection{Reduced grammars and reference counts}

We first rule out some obvious redundancies. An optimal grammar has no unreachable rules, since deleting them would reduce its size. No auxiliary rule has a single symbol on its right-hand side: if $A\to X$, writing $X$ at every reference and deleting the rule of $A$ decreases the size by one. Similarly, if $S\to A$, we may take $A$ as the start symbol instead; a start rule $S\to T$ with a single symbol is needed only when $n=1$.

Less obviously, no two nonterminals of an optimal grammar generate the same string. Such nonterminals would have expansions of equal length, so neither could depend on the other; merging their references and deleting one of the two rules would shrink the grammar without creating a cycle. Accordingly, we call a grammar \emph{reduced} if it has no unreachable rules and no single-symbol auxiliary rules, and different nonterminals generate different strings; when $|T|\ge2$, we also require that the start rule is not a single nonterminal. In a reduced grammar, every referenced nonterminal has a strictly shorter expansion than the rule in which it occurs; we will use this property frequently.

Let $r_G(A)$ denote the total number of occurrences of $A$ on all right-hand sides. Note that this counts symbol occurrences on right-hand sides, not expansions in the full derivation tree; the two numbers can differ greatly.

\begin{observation}\label{obs:twice}
In an optimal grammar, every auxiliary nonterminal is referenced at least twice.
\end{observation}

\begin{proof}
If $A$ is referenced only once, replacing this reference by $\alpha_A$ and deleting the rule of $A$ decreases the size by one, contradicting optimality.
\end{proof}

Among all grammars of minimum size, we may further choose one with the fewest rules, which we call an \emph{optimal grammar with the fewest rules}. This refinement may seem inessential, but it gives a stronger lower bound on the number of references to short rules.

\begin{lemma}\label{lem:three}
In an optimal grammar with the fewest rules, every auxiliary rule $A\to\alpha_A$ with $|\alpha_A|=2$ satisfies $r_G(A)\ge3$.
\end{lemma}

\begin{proof}
Expand all $r=r_G(A)$ references and delete the rule of $A$. The size changes by
\[
r(|\alpha_A|-1)-|\alpha_A|.
\]
For $|\alpha_A|=r=2$, this quantity is zero: the size is unchanged, while the number of rules decreases by one, contradicting the choice of the grammar. Together with Observation~\ref{obs:twice}, $r_G(A)\ge3$.
\end{proof}

In particular, an auxiliary nonterminal whose expansion has length two has two terminals on its right-hand side and is therefore covered by Lemma~\ref{lem:three}.

\subsection{Optimal parsing for a fixed dictionary}

This subsection treats the easier half of the framework, namely the case of a given dictionary. Let $\mathcal D$ be a set of distinct proper substrings of $T$, each of length at least two. For $U\in\mathcal D\cup\{T\}$, let $p_{\mathcal D}(U)$ be the minimum number of blocks needed to cut $U$ into single characters and shorter dictionary strings. Only dictionary strings $W$ with $|W|<|U|$ are allowed, so rules built from parsings depend only on shorter strings, which guarantees acyclicity. The total cost of shortest parsings is
\begin{equation}\label{eq:F}
F_T(\mathcal D)=p_{\mathcal D}(T)+\sum_{U\in\mathcal D}p_{\mathcal D}(U).
\end{equation}

For $m=|U|$, the following dynamic program computes $p_{\mathcal D}(U)=d[m]$:
\begin{equation}\label{eq:dp}
d[0]=0,\qquad
d[j]=1+\min\!\left(
d[j-1],\
\min_{\substack{W\in\mathcal D,\ |W|<m\\|W|\le j,\ U[j-|W|+1..j]=W}}
d[j-|W|]\right).
\end{equation}
If no dictionary string ends at position $j$, this reduces to $d[j]=d[j-1]+1$.

The following lemma is central to the framework: shortest parsing never underestimates the optimum, and it attains the optimum for the right dictionary.

\begin{lemma}\label{lem:parse}
For every candidate dictionary $\mathcal D$, $F_T(\mathcal D)\ge g(T)$. If $\mathcal D$ is the dictionary of some optimal grammar, then $F_T(\mathcal D)=g(T)$.
\end{lemma}

\begin{proof}
For the inequality, create one nonterminal per dictionary string, define its rule by a shortest parsing, and define the start rule by a shortest parsing of $T$. The resulting grammar generates $T$ and has size $F_T(\mathcal D)$, so $F_T(\mathcal D)\ge g(T)$.

For equality, suppose $\mathcal D$ comes from an optimal grammar. Each right-hand side of that grammar is a feasible parsing of the corresponding string, and replacing it by a shortest parsing cannot increase its length, so $F_T(\mathcal D)\le g(T)$. Together with the inequality above, this gives equality.
\end{proof}

We add two implementation remarks. If the dictionary strings have total length $O(n)$, each can be represented by its endpoints in the input, and the dynamic program can be run rule by rule in $O(n)$ words of space. Even if matches are checked character by character, the total running time remains polynomial.

\section{Disjoint representatives of optimal grammars}\label{sec:disjoint}

The previous section settled the case of a fixed dictionary, for which optimal parsings can be computed in polynomial time. The remaining problem is the large number of candidate dictionaries. In this section we prove the structural lemma (Lemma~\ref{lem:intro-disjoint}): there is a partition of the input whose blocks of length at least two form exactly the dictionary of an optimal grammar. Enumerating candidate dictionaries thus reduces to enumerating partitions of the input.

\subsection{Existence of disjoint representatives}

Traverse the derivation tree of a reduced grammar depth-first, expanding the rule of each nonterminal at its first encounter and keeping every later encounter of the same nonterminal as a leaf. This is exactly the pruning used for grammar trees in~\cite{CNP21}. After pruning, every leaf corresponds to an interval of $T$.

\begin{lemma}\label{lem:leaves}
If every auxiliary nonterminal is referenced at least twice, then in the pruned derivation tree every auxiliary nonterminal $A$ has exactly $r_G(A)-1$ unexpanded leaves.
\end{lemma}

\begin{proof}
Every reachable rule is expanded exactly once. Hence all expanded right-hand sides together contain $r_G(A)$ occurrences of $A$; exactly one of them becomes an internal node (the first encounter, which is expanded), and all others become leaves.
\end{proof}

The leaves of the pruned tree, read from left to right, spell out $T$, so their intervals are pairwise disjoint; this is the sole source of disjointness. Writing $t_G$ for the total number of terminals on all right-hand sides and adding up the lengths of the leaves, we obtain
\begin{equation}\label{eq:count}
n=t_G+\sum_{A\ne S}\bigl(r_G(A)-1\bigr)|\operatorname{val}(A)|.
\end{equation}

\begin{theorem}\label{thm:disjoint}
For every reduced grammar satisfying Observation~\ref{obs:twice}, one can choose for each dictionary string a representative occurrence in $T$ such that all representative intervals are pairwise disjoint. In particular, there is a partition $T=S_1\cdots S_k$ whose blocks of length at least two, after removing duplicates, are exactly the dictionary of an optimal grammar.
\end{theorem}

\begin{proof}
By Lemma~\ref{lem:leaves}, every auxiliary nonterminal has at least one unexpanded leaf; here $r_G(A)\ge2$ is guaranteed by Observation~\ref{obs:twice}. Choosing one leaf per auxiliary nonterminal as its representative gives pairwise disjoint intervals. If every remaining position forms a single-character block, then the blocks of length at least two are exactly the strings of the dictionary.
\end{proof}

The following small example illustrates the relation between pruning, leaves, and representatives.

\begin{example}\label{ex:running}
Consider $T=\texttt{abcabcbcbc}$ and the grammar
\[
S\to A\,A\,B\,B,\qquad A\to\texttt{a}\,B,\qquad B\to\texttt{bc}.
\]
Its size is $4+2+2=8$, with $r_G(A)=2$ and $r_G(B)=3$. After pruning, the expansions of the leaves are, in order,
\[
\texttt{a}\mid\texttt{b}\mid\texttt{c}\mid
\underbrace{\texttt{abc}}_A\mid
\underbrace{\texttt{bc}}_B\mid
\underbrace{\texttt{bc}}_B.
\]
The unexpanded leaf of $A$ lies at $[4,6]$, and the two unexpanded leaves of $B$ lie at $[7,8]$ and $[9,10]$. Taking the first two intervals gives disjoint representatives for the dictionary $\{\texttt{abc},\texttt{bc}\}$. Note that $\texttt{abc}$ and $\texttt{bc}$ have overlapping occurrences in $T$, which this choice of representatives avoids.
\end{example}

\begin{corollary}\label{cor:totallen}
A reduced optimal grammar satisfies
\begin{equation}\label{eq:totallen}
\sum_{U\in\mathcal D(G)}|U|\le n.
\end{equation}
\end{corollary}

\subsection{The partition enumeration algorithm}

By Theorem~\ref{thm:disjoint}, some block partition yields exactly the dictionary of an optimal grammar, so the algorithm is straightforward. Deciding independently for each of the $n-1$ gaps between consecutive characters of $T$ whether to cut there enumerates all $2^{n-1}$ partitions. For each partition, the distinct blocks of length at least two that are proper substrings form a candidate dictionary $\mathcal D$, and the algorithm returns the minimum of $F_T(\mathcal D)$.

\begin{theorem}\label{thm:2n}
This algorithm computes $g(T)$ in $O^*(2^n)$ time and $O(n)$ machine words of space, and it can output an optimal grammar within the same asymptotic bounds.
\end{theorem}

\begin{proof}
Correctness follows from two matching bounds: by Lemma~\ref{lem:parse}, every candidate has cost at least $g(T)$, and by Theorem~\ref{thm:disjoint}, some partition yields an optimal dictionary, whose cost is exactly $g(T)$.

Each candidate dictionary has total length at most $n$, so its cost can be computed in polynomial time and linear space, and a depth-first enumeration stores only the current partition. To output an optimal grammar, it suffices to record the endpoints of the best candidate and, after the enumeration, recompute its parsings and trace them back.
\end{proof}

\section{Linear-space exact algorithms}\label{sec:longshort}

The bound $O^*(2^n)$ already improves on $O^*(3^n)$, but base two appears to be a natural barrier for partition enumeration. The key to overcoming it is that dictionary strings of length two satisfy a constraint stronger than disjointness: in an optimal grammar with the fewest rules, each of them is referenced at least three times. This suggests treating dictionary strings differently according to their length. Long dictionary strings are enumerated by their occurrence intervals as before, whereas strings of length two are instead selected from the remaining positions by a filter. Requiring at least three occurrences makes this filter restrictive, and the resulting savings bring the base below two.

\subsection{Representative occurrences of long strings}

Fix an optimal grammar $G$ with the fewest rules and split its dictionary into
\[
\mathcal D_2=\{U\in\mathcal D(G):|U|=2\},\qquad
\mathcal D_{\ge3}=\mathcal D(G)\setminus\mathcal D_2.
\]

\begin{lemma}\label{lem:longrep}
There is a set $\mathcal I$ of pairwise disjoint intervals containing exactly one representative occurrence of each $U\in\mathcal D_{\ge3}$. Let
\[
R=[1,n]\setminus\bigcup_{I\in\mathcal I}I.
\]
Then every $U\in\mathcal D_2$ has at least three complete occurrence intervals inside $R$.
\end{lemma}

\begin{proof}
We modify the pruning of Section~\ref{sec:disjoint}: nonterminals whose expansion has length at least three are still expanded only at their first encounter, but nonterminals with expansion length two are never expanded and always kept as leaves.

Then every long rule is expanded exactly once. By Observation~\ref{obs:twice}, the corresponding nonterminal has at least one unexpanded leaf, and choosing one leaf of each such nonterminal gives $\mathcal I$.

Now consider a nonterminal $A$ of expansion length two. Its right-hand side consists of terminals only, so every reference to $A$ lies in the start rule or in a long rule. Each of these rules is expanded exactly once, so all $r_G(A)$ references to $A$ remain as leaves. By Lemma~\ref{lem:three}, $r_G(A)\ge3$. These leaves are pairwise disjoint and disjoint from the chosen long representatives, so they all lie in $R$.
\end{proof}

We use Lemma~\ref{lem:longrep} as a filter: since the short dictionary strings of an optimal grammar occur at least three times within the remaining positions, restricting to such strings loses no optimal solution. For any set $\mathcal I$ of long representatives and the corresponding set $R$ of remaining positions, define
\begin{equation}\label{eq:C2}
\mathcal C_2(R)
=
\left\{
U\in\Sigma^2:
\#\left\{
i:
i,i+1\in R,\;
T[i..i+1]=U
\right\}
\ge3
\right\}.
\end{equation}
Occurrences are counted in the original coordinates of $T$, and characters on either side of a long interval are \emph{not} considered adjacent; otherwise, spurious ``occurrences'' spanning a long representative would be counted. With this convention, Lemma~\ref{lem:longrep} directly gives $\mathcal D_2\subseteq\mathcal C_2(R)$.

\begin{observation}\label{obs:C2}
\begin{equation}\label{eq:C2bound}
|\mathcal C_2(R)|\le\left\lfloor\frac{|R|}{3}\right\rfloor.
\end{equation}
\end{observation}

\begin{proof}
Each candidate has at least three valid starting positions, each starting position determines a single string of length two, and there are at most $|R|$ valid starting positions.
\end{proof}

This factor of one third acts in the exponent; together with the exclusion of blocks of length two from the partitions, it improves the base from $2$ to $\rho<1.88$.

\subsection{The algorithm and its correctness}

The algorithm is similar to that of Section~\ref{sec:disjoint}, except that the shape of the partitions is restricted and each partition is followed by an enumeration of subsets of short strings. Specifically, we enumerate block partitions $P:T=S_1\cdots S_k$ in which every block has length one or at least three; blocks of length two are excluded because short strings are obtained by filtering. For each partition, we perform the following steps.
\begin{enumerate}[leftmargin=*]
\item Remove duplicates among the blocks of length at least three to obtain a candidate long dictionary $\mathcal L(P)$. If one of these blocks is the whole of $T$, skip the partition.
\item Let $R(P)$ be the set of positions of single-character blocks, and compute $\mathcal C_2(R(P))$.
\item For each subset $\mathcal Q\subseteq\mathcal C_2(R(P))$, compute
\begin{equation}\label{eq:candidate}
F_T\bigl(\mathcal L(P)\cup\mathcal Q\bigr).
\end{equation}
\end{enumerate}
The algorithm returns the minimum over all candidates. In the degenerate case $n\le2$ it uses the single rule $S\to T$.

\begin{lemma}\label{lem:longshort}
The algorithm above computes $g(T)$ exactly.
\end{lemma}

\begin{proof}
As before, Lemma~\ref{lem:parse} implies that every candidate has value at least $g(T)$.

Conversely, we show that the algorithm encounters an optimal solution. Take the optimal grammar of Lemma~\ref{lem:longrep}, use its long representatives as blocks, and let every remaining position form a single-character block. This partition $P$ has the required shape and is therefore enumerated, with $\mathcal L(P)=\mathcal D_{\ge3}$. Since $\mathcal D_2\subseteq\mathcal C_2(R(P))$, the subset enumeration reaches $\mathcal Q=\mathcal D_2$. The full candidate dictionary is then exactly $\mathcal D(G)$, whose cost is $g(T)$.
\end{proof}

\subsection{Running time and space}

The analysis combines the number of partitions and the number of subsets per partition into a single weighted count. If a partition $P$ has $s$ single-character blocks, then $|R(P)|=s$, and by~\eqref{eq:C2bound} it produces at most $2^{s/3}$ candidates for the short dictionary. Assigning weight $a=2^{1/3}$ to each single-character block and weight one to every other block, the total weight bounds the total number of candidates. Let $A_n$ be the total weight of all admissible partitions of length $n$. Then $A_0=1$ and
\begin{equation}\label{eq:rec}
A_n
=
aA_{n-1}
+
\sum_{\ell=3}^{n}A_{n-\ell}.
\end{equation}
By the definition of the weights, the total number of candidates is at most $A_n$. Let $\rho>1$ be the solution of
\begin{equation}\label{eq:rho}
\frac a\rho
+
\sum_{\ell=3}^{\infty}\rho^{-\ell}
=
\frac a\rho
+
\frac{1}{\rho^2(\rho-1)}
=
1.
\end{equation}
The left-hand side is continuous and strictly decreasing on $(1,\infty)$, tending to infinity as $\rho\to1^+$ and to zero as $\rho\to\infty$, so the solution exists and is unique. Induction on~\eqref{eq:rec} gives
\begin{equation}\label{eq:An}
A_n
\le
a\rho^{n-1}
+
\sum_{\ell=3}^{n}\rho^{n-\ell}
\le
\rho^n.
\end{equation}
Solving numerically gives $\rho=1.8722572196\ldots<1.88$.

The space bound is unaffected. A depth-first enumeration stores only the current partition, the long dictionary, and the candidate strings of length two, all within $O(n)$ words; since candidate dictionaries have total length $O(n)$, the dynamic program of Lemma~\ref{lem:parse} still runs in linear space. Together with~\eqref{eq:An}, this proves part~1 of Theorem~\ref{thm:main}; an optimal grammar can be output as in Theorem~\ref{thm:2n}.

\subsection{An algorithm depending on the alphabet size}\label{sec:alphalg}

Section~\ref{sec:longshort} treated dictionary strings of length two separately. Over a small alphabet, this idea extends to much longer strings: there are only $O(\sigma^L)$ strings shorter than a threshold $L$, and at most $n/L$ disjoint representatives of length at least $L$. We therefore enumerate short strings by content and long strings by position, and choose $L$ to balance the two costs.

Suppose the input uses at most $\sigma$ symbols, where $2\le\sigma\le n$; for a unary input we take $\sigma=2$. Fix an integer threshold $L\ge3$. It suffices to consider short strings that occur in $T$ as proper substrings, and the number of such candidates shorter than $L$ is at most
\[
\sum_{\ell=2}^{L-1}\sigma^\ell=O(\sigma^L),
\]
so all short dictionaries can be enumerated in $2^{O(\sigma^L)}$ time. For long strings, we again use Theorem~\ref{thm:disjoint}: the strings of length at least $L$ in an optimal dictionary have disjoint representatives, of which there are at most $\lfloor n/L\rfloor$. It remains to count the families of such disjoint intervals.

\begin{lemma}\label{lem:intervals}
The number of families of at most $n/L$ pairwise disjoint nonempty intervals of integer positions in $[1,n]$ is at most
\[
\exp\!\left(O\!\left(\frac{n\log L}{L}+\log n\right)\right).
\]
\end{lemma}

\begin{proof}
For a family of $k$ intervals, let $\ell_1,\ldots,\ell_k\ge1$ be the interval lengths in order and $g_0,\ldots,g_k\ge0$ the gap lengths. Substituting $\ell'_i=\ell_i-1$ to remove the lower bounds gives
\[
\sum_{i=1}^k\ell'_i+\sum_{i=0}^kg_i=n-k.
\]
By stars and bars, the number of families is $\binom{n+k}{2k}$. For $1\le k\le n/L$,
\[
\log\binom{n+k}{2k}
=O\bigl(k\log(en/k)\bigr)
=O(n\log L/L).
\]
Summing over all admissible $k$ adds only an $O(\log n)$ term.
\end{proof}

The algorithm is now as follows. Enumerate families of intervals, discard intervals shorter than $L$ or covering all of $T$, and remove duplicate strings to obtain a long dictionary; combine it with every short dictionary and compute the cost by~\eqref{eq:F}. Correctness follows from Lemma~\ref{lem:parse} and Theorem~\ref{thm:disjoint}, and the cost of each candidate dictionary can be computed in polynomial time. Lemma~\ref{lem:intervals} therefore gives a total running time of
\begin{equation}\label{eq:alphtime}
\exp\!\left(
O\!\left(\sigma^L+\frac{n\log L}{L}+\log n\right)
\right).
\end{equation}

It remains to choose the threshold. Let $R=\log_\sigma n$. If $R\ge6$, take $L=\max\{3,\lfloor R/2\rfloor\}$; then $\sigma^L\le\sqrt n$ and
\[
\frac{n\log L}{L}
=O\!\left(\frac{n\log(R+1)}{R}\right).
\]
Since $R\le\log_2n$, we have $\sqrt n=o(n/R)$ for sufficiently large $n$. Substituting these bounds into~\eqref{eq:alphtime} yields
\begin{equation}\label{eq:alphbound}
2^{O\!\left(\frac{n}{R}\log(R+1)\right)}.
\end{equation}
If $R<6$, we use the algorithm of Section~\ref{sec:longshort}; since $\log(R+1)/R=\Theta(1)$ in this case, the bound~\eqref{eq:alphbound} holds as well.

The space usage is again linear. When $R\ge6$, the long representatives have total length at most $n$, and for short strings we store each string that occurs in $T$ together with one of its positions, using
\[
O(L\sigma^L)=O(R\sqrt n)=O(n)
\]
space. Hence the algorithm uses $O(n)$ machine words, which proves part~2 of Theorem~\ref{thm:main}. In particular, for a fixed alphabet the running time is $\exp(O(n\log\log n/\log n))$, and more generally it is subexponential whenever $\sigma=n^{o(1)}$.

\section{Sparse interval certificates and a uniform alphabet-sensitive algorithm}
\label{sec:sparse-interval-tensor}

The disjoint-representative theorem gives a certificate in which all rule
definitions and references are input intervals. We represent the remaining
short-interval choices by a tensor network and exploit a bound on the number
of selected intervals. The alphabet enters through an upper bound on the
optimum grammar size, rather than through a separate enumeration of short
dictionary contents. We assume $n\ge2$ throughout this section; an input of
length one has optimum size one and is handled directly.

\subsection{External counting input and the parameterized bound}

The external ingredient is the Boolean tensor contraction theorem in the
quantitative appendix of the OpenAI preprint~\cite{sat}, in the version
at commit
\texttt{42907c9e999cc00d6abe7a50bfc1023250edd1b9}. We state it in the form
needed here, with all of its hypotheses. Fix absolute positive constants $C,A,b$ and let
$Q\ge2$ be an integer. A bipartite network with Boolean wires, at most $CQ$
wires, maximum degree at most $\lceil Q^{9/10}\rceil$, and a number of vertices
bounded by a fixed polynomial in $Q$ has entries in
\[
 \mathcal B=(\mathbb Z/2^P\mathbb Z)[z]/(z^\ell),
 \qquad 1\le P,\ell\le CQ.
\]
Let $\mathcal L$ be the total length of its graph and tensor descriptions.
If each tensor entry is computable uniformly from its description in
$(Q+\mathcal L)^b2^{A Q^{9/10}}$ time and $(Q+\mathcal L)^b$ space, its contraction
can be computed deterministically in
\[
 (Q+\mathcal L)^{O(1)}2^{O(Q/\sqrt{\log Q})}
\]
time and polynomial space. All entries below are computable in polynomial
time. We use this tensor theorem directly, together with its bounds on wires,
degree, and precision; no fixed-width SAT algorithm is used as a black box.

For such a network, an assignment of a bit to every wire contributes the
product of the incident tensor entries. The contraction is the sum of these
products over all assignments. Parallel wires are distinct. A tensor with
no incident wire contributes its scalar entry.

\begin{theorem}[Budget and threshold parameters]\label{thm:sparse-budget}
Put $t=\log_2(n+2)$. Suppose a known integer $K$ satisfies
\[
 g(T)\le K\le n,\qquad K\ge n/t.
\]
For any integer $2\le L\le t^2$ with $L<n$, an optimal grammar for $T$ can
be found deterministically in polynomial space and time
\begin{equation}\label{eq:sparse-budget-time}
 n^{O(1)}2^{O\left(
 \frac{n\log(L+1)}L+
 \frac{K\log(\mathrm e nL/K)}{\sqrt{\log n}}
 \right)}.
\end{equation}
All constants are absolute. Bounded exceptional values of $n$ are handled
directly throughout this section.
\end{theorem}

\subsection{A certificate with three interval families}

Use three families $D,R_0,R_1$ of nonempty intervals of $T$ with the following
properties. The dictionary representatives $D$ are pairwise disjoint, have
lengths between two and $n-1$, and have pairwise distinct contents. The
intervals of $R_0$ partition $[1,n]$. The intervals of $R_1$ partition the
union of $D$, and each is contained in one member of $D$ and strictly shorter
than that member. Every reference interval in $R_0\cup R_1$ of length at
least two has the content of a dictionary interval. An $R_0$ reference is
also strictly shorter than $n$.

\begin{lemma}\label{lem:sparse-certificate}
These conditions define a grammar of size $|R_0|+|R_1|$, and some certificate
has size $g(T)$.
\end{lemma}
\begin{proof}
Create a nonterminal for each dictionary content. Read the start production
from $R_0$ and the defining production for a dictionary interval from the
$R_1$ intervals contained in it. A length-one interval contributes its
terminal letter; a longer one contributes its dictionary nonterminal.
Every reference in a definition strictly decreases expansion length, so
the dependency graph is acyclic. Induction on expansion length shows that
each rule generates its specified string. The start expansion is $T$ and
the sum of right-hand-side lengths is $|R_0|+|R_1|$.

Conversely, choose a reduced optimal grammar and its pairwise disjoint
representatives from Theorem~\ref{thm:disjoint}. Transfer each defining
right-hand-side parsing to its representative occurrence and the start
parsing to $T$. These give the required families with unchanged cost.
\end{proof}

These conditions do not require every dictionary entry to be referenced.
An unreferenced entry is harmless: its definition is still counted in
$R_1$, and the certificate still yields a grammar within the budget.

Fix $L$. Call intervals of length greater than $L$ long, and enumerate
\[
 D^{>L},\qquad R_0^{>L},\qquad R_1^{>L}.
\]
Each family is pairwise disjoint and has total length at most $n$. A family
of $q\le n/(L+1)$ long intervals has at most
$\binom{n+2q}{2q}$ possible ordered endpoint sequences. Summing over $q$
and taking three families gives at most
\begin{equation}\label{eq:sparse-patterns}
 n^{O(1)}2^{O(n\log(L+1)/L)}
\end{equation}
patterns, enumerable one at a time in polynomial space.

Reject a pattern unless its long dictionary contents are distinct proper
substrings, every long reference matches a long dictionary content, every
long $R_1$ reference is contained in a long dictionary representative and
strictly shorter than it, and every long $R_0$ reference is shorter than $T$.
A short dictionary definition cannot contain a long reference, so these
checks cover the parent relation of every long reference. The remaining
coverage conditions are enforced by the network below.

For a decision budget $b\le K$, let
\[
 c_{\rm long}=|R_0^{>L}|+|R_1^{>L}|.
\]
If $c_{\rm long}>b$, reject this pattern. Otherwise the short references
have budget $b'=b-c_{\rm long}$.

\subsection{The interval-selection network}

For each short input interval, create separate Boolean choices for its
possible roles as a dictionary representative, an $R_0$ token, or an $R_1$
token. Dictionary choices have lengths $2,\ldots,L$; token choices have
lengths $1,\ldots,L$. There are $O(nL)$ choices in total. The long pattern
is fixed data and contributes no wires.

Partition the input positions into consecutive blocks of width
$B=\lceil\sqrt n\rceil$, apart from a possible shorter last block. Since
$L\le t^2$, sufficiently large $n$ satisfy $L\le B$. Every short interval
then meets at most two position blocks.

\paragraph{Selection tensors.}
Each candidate choice has one port to each position block it meets and,
if its length is at least two, one port to its content constraint. The
degree of this tensor is at most three. Its entry is zero unless all ports
carry the same bit. When they all carry zero its entry is one. When they
all carry one its entry is $z$ for an $R_0$ or $R_1$ token, and one for a
dictionary representative. Thus every selected short right-hand-side token
contributes exactly one factor $z$. The definition of a selected dictionary
entry is charged through its $R_1$ tokens.

\paragraph{Position tensors.}
A position block receives the choice bits of every short interval meeting
it, together with the fixed long pattern. At every position of that block
it requires the dictionary coverage to be at most one, the $R_0$ coverage
to be exactly one, and the $R_1$ coverage to equal the dictionary coverage.
It also checks every selected short $R_1$ token whose first position lies
in the block: the token must be wholly contained in its unique dictionary
parent and strictly shorter than that parent.

The last check is local. Any selected short parent covering the token's
first position meets this block, so its choice bit is present; long parents
are fixed data. The endpoints and lengths of every candidate are known.
This check is necessary: equal coverage counts alone would allow a
token to cross the boundary between adjacent dictionary representatives.
The degree of a position tensor is
\[
 O(BL+L^2)=O(\sqrt n\,L),
\]
and its entry is computable in polynomial time.

\paragraph{Content tensors.}
For a distinct short string $w$ of length at least two, collect all candidate
ports with content $w$. Let $d_w$ count its selected dictionary ports and
let $r_w$ indicate whether some reference port is selected. Require
\begin{equation}\label{eq:sparse-content}
 d_w\le1,\qquad r_w=1\ \Longrightarrow\ d_w=1.
\end{equation}
This enforces distinct dictionary contents and supplies a definition for
every short reference. A short content cannot equal a long dictionary
content, so no cross-threshold case is missing.

There are at most $3n$ ports for one content. If their number is at most
$B$, use a single content tensor. Otherwise partition the ports into
blocks of at most $B$ inputs. Each content-block tensor rejects two or more
selected dictionary inputs and has two additional ports carrying
\[
 \delta=\mathbf1[\text{one dictionary input is selected}],\qquad
 \rho=\mathbf1[\text{some reference input is selected}].
\]
Both output bits are uniquely determined by the inputs. A root tensor
requires the sum of the $\delta$ bits to be at most one and to equal one
whenever some $\rho$ bit is one. In particular, an all-zero input block has
summary $(0,0)$.

The number of added summary wires is $O(nL)$ and their root degrees are
$O(n/B)$. The network is bipartite: place the selection tensors and the
content roots on one side, and position tensors, content blocks, and
unsplit content tensors on the other. Writing $M$ for its number of wires,
we have
\begin{equation}\label{eq:sparse-raw-size}
 M=O(nL),\qquad \Delta_0=O(\sqrt n\,L).
\end{equation}
Every local entry is zero, one, or $z$, and is computable in polynomial
time. Before budget truncation, nonzero configurations are precisely the interval certificates
extending the fixed long pattern, with the exponent of $z$ equal to their
number of short right-hand-side tokens.

\subsection{Sparsity of budget-feasible configurations}

\begin{lemma}\label{lem:sparse-activity}
Every nonzero configuration of short cost at most $b'\le K$ has at most
$8K$ original wires equal to one, including all content-summary wires.
\end{lemma}
\begin{proof}
Let $r$ be its number of selected short reference tokens and $d$ its number
of selected short dictionary representatives. Each short dictionary rule
is tiled by strictly shorter tokens and hence has at least two tokens, all
short. The representatives are disjoint, so these defining tokens are
distinct. Consequently $2d\le r$ and
\[
 A:=d+r\le\tfrac32r\le\tfrac32K.
\]
The selected candidates have at most $3A$ incident wires in total. A
content block with a nonzero summary contains at least one selected
candidate. Its input blocks are disjoint, and a candidate has at most one
content port, so at most $A$ blocks have nonzero summaries. Their two
summary ports contribute at most $2A$ further ones. These are all the types of
wires, so at most $5A\le15K/2<8K$ wires equal one.
\end{proof}

\subsection{Indexing selected wires instead of all choices}

Put $H=8K$ and $J=\lceil n^{1/4}\rceil$. On each side of the bipartition,
assign vertices greedily to $J$ groups, always choosing a group of least
total incident degree. Its final load is at most
\begin{equation}\label{eq:sparse-group-load}
 M/J+\Delta_0=O(n^{3/4}L).
\end{equation}
Let $U_{gh}$ be the number of original wires joining left group $g$ to
right group $h$. Enumerate the arrays
\[
 0\le a_{gh}\le U_{gh},\qquad \sum_{g,h}a_{gh}\le H.
\]
There are at most
\begin{equation}\label{eq:sparse-arrays}
 (H+1)^{J^2}=2^{O(\sqrt n\log n)}
\end{equation}
arrays, processed sequentially.

For a fixed array, replace the $U_{gh}$ original wires by
\[
 w_{gh}=\left\lceil\log_2\binom{U_{gh}}{a_{gh}}\right\rceil
\]
Boolean wires indexing the subset of original wires that are one. Every
other original wire is zero. Index subsets in lexicographic order and set
an entry to zero on any invalid index. Ranking and unranking use binomial
coefficients and polynomial time and space. When the binomial coefficient
equals one, no wire is needed: both groups know the unique subset.

A new group tensor decodes all its incident subset indices and multiplies
the corresponding original tensor entries. This is an exact re-encoding:
every original wire assignment of weight at most $H$ determines one array
and one valid assignment of its index wires, and conversely. Each original
tensor appears in exactly one group.

The number $W$ of new wires satisfies
\begin{align}
 W&\le J^2+\log_2\prod_{g,h}\binom{U_{gh}}{a_{gh}}\nonumber\\
  &\le J^2+\log_2\binom{M+H}{H}
   =O\left(K\log\frac{\mathrm e nL}{K}\right).
   \label{eq:sparse-wire-bound}
\end{align}
Indeed, a choice of all subsets determines a global subset of at most $H$
wires, and
$\sum_{j=0}^{\min(H,M)}\binom Mj\le\binom{M+H}H$ by Vandermonde's identity.
The last estimate uses $M=O(nL)$, $H=8K$, and $J^2=O(\sqrt n)=O(K)$.

Set
\[
 Q=\left\lceil K\log_2(\mathrm e nL/K)\right\rceil.
\]
Since $w_{gh}\le U_{gh}$, the maximum group degree is at most the load
in~\eqref{eq:sparse-group-load}. Uniformly for the allowed parameters,
\[
 Q\ge K\ge n/t,\qquad L\le t^2,\qquad
 O(n^{3/4}L)\le Q^{9/10}
\]
for all sufficiently large $n$: the ratio is
$O(n^{-3/20}t^{29/10})=o(1)$. Also $\log Q=\Theta(\log n)$ uniformly.
The group tensors have polynomial-size descriptions, and evaluating an
entry only decodes subsets and computes the original local checks and
their product. Their entries are therefore computable in polynomial time
and space.

\subsection{Budget truncation and exact positive counting}

Choose the precision
\begin{equation}\label{eq:sparse-precision}
 P=1+\left\lceil\log_2\binom{M+H}{H}\right\rceil=O(Q)
\end{equation}
and work in
\[
 \mathcal B=(\mathbb Z/2^P\mathbb Z)[z]/(z^{b'+1}).
\]
Its truncation order $b'+1\le K+1$ is also $O(Q)$. A valid original
configuration of short cost $r$ contributes $z^r$, so costs above $b'$
vanish. By Lemma~\ref{lem:sparse-activity}, no budget-feasible contribution
is lost by restricting the original Hamming weight to at most $H$.

The tensor-contraction theorem applies to each fixed array. Its time is
\[
 n^{O(1)}2^{O(Q/\sqrt{\log n})}.
\]
Array enumeration is absorbed in this exponent, since
\[
 \sqrt n\log n=o\bigl((n/t)/\sqrt{\log n}\bigr)
             =o\bigl(Q/\sqrt{\log n}\bigr).
\]
After summing the contractions of all arrays in $\mathcal B$, sum the
remaining coefficients. The coefficients are summed only
after all multiplications and truncations; this final summation is not a
ring homomorphism, and we do not use it as one.

The result is the number of budget-feasible certificates modulo $2^P$.
That number is at most the number of original Boolean wire configurations
of weight at most $H$, which is bounded by
\[
 \sum_{j=0}^{\min(H,M)}\binom Mj
 \le\binom{M+H}H<2^P.
\]
Its least nonnegative residue is therefore the exact count. In particular,
a positive count cannot be mistaken for zero modulo $2^P$. Different
certificates may describe the same grammar; only existence is required.
The precision in~\eqref{eq:sparse-precision} is accounted for in the
complexity analysis; no unbounded-precision arithmetic is assumed.

Combining this per-pattern time with~\eqref{eq:sparse-patterns} proves the
decision bound in~\eqref{eq:sparse-budget-time}.

\subsection{Recovering an optimum grammar}

Binary search on $b\in[1,K]$ finds the optimum value. At that budget, retain
one long pattern whose count is positive. Fix the short candidate choices
successively: test whether the count remains positive when the next
selection tensor is restricted to bit zero; if not, restrict it to bit one.
Restrictions multiply the relevant entries by zero or one and introduce
no new wires. All degree, sparsity, precision, and description-size bounds
remain valid. Since there are $O(nL)$ choices, this adds only a polynomial
factor. Read the resulting grammar from $D,R_0,R_1$ as in
Lemma~\ref{lem:sparse-certificate}.

Every stage uses polynomial space. Patterns and arrays are processed one at
a time; indices are decoded individually; a ring element has at most $K+1$
coefficients of $O(Q)$ bits each; and the external contraction routine uses
polynomial space. This completes the proof of
Theorem~\ref{thm:sparse-budget}.

\subsection{Uniform dependence on the alphabet}

\begin{lemma}\label{lem:sparse-size-bound}
For $n\ge2$, $s=\max\{2,|\operatorname{alph}(T)|\}$ and $R=\log_s n$, one has
$g(T)\le\min\{n,10n/R\}$.
\end{lemma}
\begin{proof}
If $R<8$, use the literal rule. Otherwise let $a=\lfloor R/4\rfloor$ and
split $T$ into full length-$a$ blocks and a tail. Define each distinct
full block literally, and let the start rule list their names and the tail.
The size is at most
\[
 n/a+a+a s^a\le 8n/R+R/4+(R/4)n^{1/4}\le10n/R.
\]
Here $R\le\log_2 n$ and $n\ge256$; in this range
$(\log_2 n)^2\le n^{3/4}$. Taking the better of this construction and the
literal rule proves the claim.
\end{proof}

\begin{theorem}\label{thm:sparse-uniform}
With $s$ and $R$ as above, an optimal grammar can be found deterministically
in polynomial space and time
\begin{equation}\label{eq:sparse-uniform}
 n^{O(1)}2^{O\left(
 \frac{n\log\log(n+2)}{R\sqrt{\log(n+2)}}
 \right)}.
\end{equation}
The constants are independent of $s$.
\end{theorem}
\begin{proof}
Put $t=\log_2(n+2)$ and take $K=\min\{n,\lceil10n/R\rceil\}$. Then $g(T)\le K$,
$K=\Theta(n/R)$, and $K\ge n/t$. Set $L=\lceil R\sqrt t\rceil$.
For sufficiently large $n$ this satisfies $2\le L\le t^2$ and $L<n$.
The two terms in~\eqref{eq:sparse-budget-time} are bounded by
\begin{align*}
 \frac{n\log(L+1)}L
 &=O\left(\frac{n\log t}{R\sqrt t}\right),\\
 \frac{K\log(\mathrm e nL/K)}{\sqrt t}
 &=O\left(\frac{n\log(\mathrm e RL)}{R\sqrt t}\right)
  =O\left(\frac{n\log t}{R\sqrt t}\right),
\end{align*}
uniformly for $1\le R\le\log_2 n$. Apply
Theorem~\ref{thm:sparse-budget}. For an alphabet upper bound
$2\le\sigma\le n$, replace $s$ by $\sigma$; this can only weaken the bound.
\end{proof}

Theorem~\ref{thm:sparse-uniform} is Theorem~\ref{thm:subexp-intro}.

\section{Binary encodings and NP-completeness}\label{sec:alphabet}

\subsection{Binary codewords and anchors}\label{sec:anchors}

Replace each symbol of the source alphabet by a binary codeword of fixed length, and denote the resulting string by $C(W)$. One direction is easy: a grammar for $W$ becomes a grammar for $C(W)$ once a rule is added for each codeword. The converse is harder, since a grammar for $C(W)$ may refer to fragments inside codewords, or even to fragments spanning adjacent codewords, and such references correspond to no rule over $W$. To handle them, we designate in each codeword a position that represents the source symbol.

This position is called an \emph{anchor}, and its \emph{label} is the corresponding source symbol. If a reference covers at most one anchor, its labels form the empty string or a single symbol, which poses no problem for the translation. The difficulty lies in references covering several anchors: such a fragment can be represented by a single nonterminal of the source grammar only if all its occurrences yield the same label sequence. The goal of this subsection is to design an encoding that guarantees this consistency.

Let the code length $m$ be a multiple of four, and place the anchor at bit $m/2$ of each codeword. For adjacent source symbols $a,b$, define the boundary fragment
\begin{equation}\label{eq:boundary}
F(a,b)=\bigl(\text{the last }m/4\text{ bits of }C(a)\bigr)\;\Vert\;\bigl(\text{the first }m/4\text{ bits of }C(b)\bigr).
\end{equation}
We require the encoding to satisfy the following property, which makes boundary fragments unambiguous alignment markers: every occurrence of a boundary fragment uniquely determines the position of a codeword boundary and the labels on both of its sides.

\begin{syncprop}
If a window of length $m/2$ in a concatenation of codewords equals $F(a,b)$, then the window starts at offset $3m/4$ inside a codeword and spans the true boundary between $C(a)$ and $C(b)$.
\end{syncprop}

Since a boundary fragment lies between any two adjacent anchors, every substring covering two adjacent anchors contains the complete fragment between them.

\begin{lemma}\label{lem:code}
For every $q\ge2$ there is a binary fixed-length code satisfying (SI) with code length $O(\log q)$. Specifically, every multiple of four with $m\ge64$ and $m>32\log_2q$ is a valid code length. For $q\le24$, the code length $m=128$ suffices.
\end{lemma}

\begin{proof}
We use the probabilistic method. Choose all $mq$ codeword bits independently and uniformly at random, and fix a boundary fragment $F(a,b)$ and a window
\[
Y=(C(c)C(d))[r+1..r+m/2],\qquad
0\le r<m,\qquad
(c,d,r)\ne(a,b,3m/4).
\]
We show that
\begin{equation}\label{eq:collide}
\Pr[F(a,b)=Y]\le 2^{-m/4}.
\end{equation}

There are two cases. If $r=3m/4$, then $Y=F(c,d)$ with $(a,b)\ne(c,d)$. If $a\ne c$, equality of the left halves requires the last $m/4$ bits of two different codewords to coincide, which happens with probability exactly $2^{-m/4}$; if $a=c$, then $b\ne d$ and the same computation applies to the right halves.

If $r\ne3m/4$, the window is misaligned. Consider the graph whose vertices are the random bits involved in the comparison and whose edges are the bitwise equality constraints. Since the phases differ, there are no self-loops, and since each bit occurs at most once in a single window, the graph has maximum degree at most two; hence each connected component is a path or a cycle. A path with $e$ edges yields $e$ independent constraints, and a cycle with $e\ge2$ edges yields $e-1\ge e/2$. Hence the $m/2$ edges provide at least $m/4$ independent constraints, and~\eqref{eq:collide} holds as well.

There are fewer than $q^4m$ bad pairs of a boundary fragment and a window, so the union bound gives
\begin{equation}\label{eq:union}
\Pr[\text{(SI) fails}]<q^4\,m\,2^{-m/4}.
\end{equation}
For $m\ge64$ we have $\log_2m\le m/8$, so $m>32\log_2q$ suffices to make the right-hand side of~\eqref{eq:union} less than one, which proves existence with $m=O(\log q)$. For the concrete case $q\le24$ and $m=128$, direct computation gives
\[
24^4\cdot128\cdot2^{-32}=0.0098876953125<1.\qedhere
\]
\end{proof}

For a fixed alphabet, a suitable code table can thus be hard-wired into the reduction, making it deterministic; Appendix~\ref{app:code} gives an explicit table for $q=24$ and $m=128$, together with a verification program.

Now fix a code satisfying (SI) and let
\[
W=w_1\cdots w_n,\qquad T=C(W).
\]
The $i$-th anchor is at position $p_i=m(i-1)+m/2$ and has label $w_i$. For an interval $I$ of the target string, let $\pi(I)$ be the string of labels of the anchors in $I$, read in order. For adjacent intervals,
\begin{equation}\label{eq:additive}
\pi(IJ)=\pi(I)\pi(J).
\end{equation}
The boundary fragment between adjacent anchors $p_i$ and $p_{i+1}$ occupies the interval
\begin{equation}\label{eq:fragpos}
\Bigl[\,m(i-1)+\tfrac{3m}4+1,\;\; mi+\tfrac m4\,\Bigr],
\end{equation}
whose endpoints both lie strictly between the two anchors. The following lemma, which is the key result of this subsection, establishes the consistency discussed above.

\begin{lemma}\label{lem:consistent}
Let a factor $X$ occur in $T=C(W)$ at interval $I$. If $I$ contains at least two anchors, then every other occurrence $I'$ of $X$ satisfies $\pi(I')=\pi(I)$. Moreover, the anchors have the same relative positions and labels in both occurrences.
\end{lemma}

\begin{proof}
Let the anchors in $I$ be $p_i,\ldots,p_j$ with $j\ge i+1$. The boundary fragment $F(w_i,w_{i+1})$ between the first two anchors lies entirely in $I$. Since $T[I']=T[I]$, the same boundary fragment appears at the corresponding offset in $I'$.

By (SI), this window in $I'$ also spans a true codeword boundary at the same phase, with labels $w_i$ and $w_{i+1}$ on its two sides. Since all codewords have the same length, all boundaries and anchors of the two occurrences are aligned. Applying (SI) to each pair of adjacent anchors in turn determines all labels, so $\pi(I')=\pi(I)$.
\end{proof}

\begin{lemma}\label{lem:project}
For every binary grammar $H$ generating $T=C(W)$, there is a grammar $H^\downarrow$ over the source alphabet generating $W$ with
\begin{equation}\label{eq:project}
|H^\downarrow|\le |H|.
\end{equation}
\end{lemma}

\begin{proof}
If $|W|=1$, a single terminal rule suffices. Assume $|W|\ge2$ and delete the unreachable rules of $H$.

We divide the nonterminals into two classes. A nonterminal $A$ is \emph{multi-anchor} if some occurrence of it contains at least two anchors; by Lemma~\ref{lem:consistent}, all its occurrences have the same projection, which we denote by $D(A)$. All other nonterminals are \emph{local}: each of their occurrences has a projection of length at most one, but different occurrences may have different projections, so local nonterminals cannot be translated consistently and are deleted instead. The start symbol $S$ is multi-anchor, with $D(S)=W$.

For each multi-anchor nonterminal $A$, we create a nonterminal $A^\downarrow$ as follows. Pick any occurrence $I_A$ of $A$. If the rule is $A\to X_1\cdots X_t$, the right-hand side divides this occurrence into adjacent intervals $I_1,\ldots,I_t$. We build the new right-hand side symbol by symbol:
\begin{itemize}
\item if $X_j=B$ is a multi-anchor nonterminal, write $B^\downarrow$;
\item if $X_j$ is a local nonterminal or a terminal, write $\pi(I_j)$, which is a single source symbol or empty.
\end{itemize}
Let $\beta_A$ be the resulting right-hand side and define $A^\downarrow\to\beta_A$. By induction along the dependencies of the original grammar, multi-anchor children generate $D(B)=\pi(I_j)$ and the other symbols generate $\pi(I_j)$ directly. By the additivity~\eqref{eq:additive},
\[
\operatorname{val}(A^\downarrow)
=\pi(I_1)\cdots\pi(I_t)
=\pi(I_A)=D(A).
\]
All new dependencies are inherited from the original grammar, so the new grammar is acyclic, and its start symbol generates $W$. Each symbol on a right-hand side is replaced by at most one symbol, so $|\beta_A|\le|\alpha_A|$, and the rules of local nonterminals are deleted. Hence $|H^\downarrow|\le|H|$.
\end{proof}

The proof uses only one property of the encoding: if some occurrence of a fragment contains at least two anchors, then all occurrences of that fragment yield the same labels. Isolating this property gives the form needed in Section~\ref{sec:eth}, which in addition quantifies the savings from deleting local rules.

\begin{projthm}
Designate some labeled anchors in a target string $X$, and let $W$ be the string of their labels read in order, with $|W|\ge2$. Suppose every factor satisfies the following: if one occurrence $I$ of it contains at least two anchors, then every other occurrence $I'$ satisfies $\pi(I')=\pi(I)$. Then for every straight-line grammar $H$ generating $X$ there is a grammar $H^\downarrow$ generating $W$ with
\begin{equation}\label{eq:projthm}
|H^\downarrow|\le |H| .
\end{equation}
If $\mathcal L$ is the set of local nonterminals, each of whose occurrences contains at most one anchor, then moreover
\begin{equation}\label{eq:projthm2}
|H^\downarrow|
\le |H|-\sum_{A\in\mathcal L}|\alpha_A| .
\end{equation}
\end{projthm}

\begin{proof}
Apply the construction of Lemma~\ref{lem:project}. Uniqueness of projections makes the translation of multi-anchor nonterminals consistent, and since local rules are deleted, their right-hand side lengths can be subtracted from the size of the original grammar, which gives~\eqref{eq:projthm2}.
\end{proof}

\begin{theorem}\label{thm:encode}
Let $\Sigma$ be a fixed alphabet with $q=|\Sigma|\ge2$, and use the synchronizing code of Lemma~\ref{lem:code} with code length $m=O(\log q)$. For every nonempty source string $W$,
\begin{equation}\label{eq:encode}
g_\Sigma(W)
\le
g_2(C(W))
\le
g_\Sigma(W)+mq .
\end{equation}
In particular, for $q\le24$ we may take $m=128$, and the additive error is at most $3072$.
\end{theorem}

\begin{proof}
The left inequality follows by applying Lemma~\ref{lem:project} to an optimal target grammar. For the right inequality, start from an optimal grammar for $W$, add a rule $Q_a\to C(a)$ for every source symbol $a$ used, and replace $a$ by $Q_a$ on the existing right-hand sides. The old rules keep their lengths, the new rules have total length at most $mq$, and the resulting grammar generates $C(W)$.
\end{proof}

For the fixed-alphabet reduction, a constant additive error suffices: taking a constant number of copies of the source graph makes the decision gap exceed it.

\subsection{NP-completeness over a binary alphabet}\label{sec:npc}

We now combine Theorem~\ref{thm:encode} with the $24$-letter reduction of Casel et al.~\cite[Section~3]{CFGGS21}, which maps a graph $G$ of maximum degree at most three to a string $w_G$. If $G$ has $N$ vertices and $M$ edges, then~\cite[Lemma~8]{CFGGS21} gives
\begin{equation}\label{eq:casel}
g(w_G)=299N+3M+5+\tau(G).
\end{equation}
The size measure of~\cite{CFGGS21} counts the axiom string together with the auxiliary rules; viewing the axiom as the right-hand side of the start rule leaves the size unchanged, so the result applies directly. The output has length $\Theta(N\log N)$; see Appendix~\ref{app:fixed}.

One obstacle remains: the vertex cover term in~\eqref{eq:casel} has coefficient one, whereas the additive error of Theorem~\ref{thm:encode} is $D=3072$, which would obscure the decision threshold. We therefore first take $K=D+1$ copies of the source graph, which amplifies the unit gap in the vertex cover size by a factor of $K$, so that it exceeds the error. Using the $24$-letter code of Theorem~\ref{thm:encode}, let
\[
D=24\cdot128=3072,\qquad K=D+1=3073.
\]
Given a Vertex Cover instance $(G,\gamma)$, let $H$ be the disjoint union of $K$ copies of $G$. Its maximum degree is still at most three, and
\[
|V(H)|=KN,\qquad |E(H)|=KM,\qquad \tau(H)=K\tau(G).
\]
Apply the reduction above to $H$ and then the binary encoding, obtaining $T_G=C(w_H)$. Denote the part of~\eqref{eq:casel} other than the vertex cover term by
\[
B(G)=299KN+3KM+5.
\]
By~\eqref{eq:casel} and Theorem~\ref{thm:encode},
\begin{equation}\label{eq:sandwich}
B(G)+K\,\tau(G)
\ \le\ g_2(T_G)
\ \le\ B(G)+K\,\tau(G)+D.
\end{equation}
Let $k_G=B(G)+K\gamma+D$. If $\tau(G)\le\gamma$, the right-hand inequality directly gives $g_2(T_G)\le k_G$. Conversely, if $\tau(G)>\gamma$, integrality gives $\tau(G)\ge\gamma+1$, and so
\[
g_2(T_G)\ge B(G)+K(\gamma+1)=k_G+(K-D)=k_G+1.
\]
Since $K-D=1$, the two cases are separated. In summary,
\begin{equation}\label{eq:iff}
\tau(G)\le\gamma
\quad\Longleftrightarrow\quad
g_2(T_G)\le k_G.
\end{equation}

\begin{theorem}[Theorem~\ref{cor:binary} restated]\label{thm:npc}
The decision version of SGP over a binary alphabet is NP-complete.
\end{theorem}

\begin{proof}
For NP-hardness, note that $K$ and the code length are constants, so $(T_G,k_G)$ can be constructed in polynomial time, and~\eqref{eq:iff} shows that this is a reduction from Vertex Cover on graphs of maximum degree at most three to binary SGP.

Membership in NP requires some care, since the expansions of a grammar can be exponentially long. We may assume that $k\le n$, since otherwise the single rule $S\to T$ suffices. A grammar of size at most $n$ has at most $n$ nonterminals and can be encoded in $O(n\log n)$ bits. The verifier checks the format of the rules and the acyclicity of the dependencies, and then computes all expansion lengths in topological order; after deleting unreachable rules, it rejects as soon as some expansion length exceeds $n$. Otherwise, all expansions have length at most $n$, so they can be computed explicitly in polynomial time, and the expansion of the start symbol is compared with the input.
\end{proof}

Since the number $K$ of copies and the code length $m$ are constants, $|T_G|=\Theta(N\log N)$. This gives the following unconditional transfer.

\begin{corollary}\label{cor:ethbinary}
If binary SGP is decidable within time $\mathcal T(n,2)$, with $\mathcal T$ nondecreasing in its length argument, then Vertex Cover on an $N$-vertex graph of maximum degree three is decidable within $\operatorname{poly}(N)+\mathcal T(CN\log(N+2),2)$ for an absolute constant $C$. The same follows from any algorithm that accepts all strings over an alphabet containing two symbols.
\end{corollary}

Section~\ref{sec:transfer} composes this reduction with a sparse-SAT reduction and states explicit logarithmic speed scales. Binary NP-completeness does not depend on these quantitative statements.

\section{Reductions with tunable alphabet size}\label{sec:eth}

The algorithms of Theorems~\ref{thm:subexp-intro} and~\ref{thm:main} depend continuously on the alphabet size. To relate improvements at different alphabet sizes to graph and SAT algorithms, this section proves Theorem~\ref{thm:tunable}, which gives reductions with a tunable alphabet size.

We first explain why the strategy of Section~\ref{sec:npc} fails here. There, copies of the source graph amplified the unit gap in the vertex cover size, which worked because the cost $mq$ of the code table was constant. If $q$ grows with the input, so does the error bound $mq$; the required number of copies is then no longer constant, and the resulting increase in the output length destroys the scale $n/R=\Theta(N)$ needed for quantitative transfers.

We therefore take a different approach: instead of dominating the error, we compute it exactly. To this end, we prepend to the source string a prefix in which every codeword occurs exactly seven times, which forces an optimal grammar to create a rule for every codeword. The argument has three steps. We first determine all fragments of the source string that can repeat, then determine which repetitions may cross codeword boundaries after encoding, and finally prove that an optimal grammar defines all codewords. The cost of the code table can then be subtracted exactly, rather than merely bounded, when translating back to the source grammar.

\subsection{Overview of the reduction}

Let $G$ be a simple graph with a nonempty edge set and maximum degree at most three, with $N$ vertices and $M$ edges. Fix an integer $2\le b\le2N$. We construct $T_{G,b}$ with
\begin{equation}\label{eq:TGb}
|\operatorname{alph}(T_{G,b})|=2b+26,
\qquad
|T_{G,b}|=\Theta\!\bigl(N(1+\log_bN)\bigr),
\end{equation}
whose optimal grammar size is
\begin{equation}\label{eq:gTGb}
\boxed{g(T_{G,b})=379N+3M+5+\tau(G).}
\end{equation}
The constants hidden in the length bound do not depend on $b$. This is Theorem~\ref{thm:tunable}; its proof occupies the rest of this section. Crucially,~\eqref{eq:gTGb} is an equality: since $379N+3M+5$ can be computed directly, $\tau(G)$ can be read off from $g(T_{G,b})$.

\subsection{The cost of codeword rules}

We first bound the cost that an optimal grammar must pay for defining these codewords. The answer is three symbols per codeword, even if the grammar uses rules with two symbols. Suppose a reduced grammar contains $r$ distinct auxiliary nonterminals $P_i$ with
\[
\operatorname{val}(P_i)=u_i\star u_i^R,\qquad u_i\ne\varepsilon,
\]
where $u_i$ does not contain $\star$. We call these strings \emph{complete codewords}.

\begin{lemma}\label{lem:codecost}
The rules of the complete codewords, together with certain other auxiliary rules whose expansions contain exactly one $\star$, have total right-hand side length at least $3r$.
\end{lemma}

\begin{proof}
Suppose $t$ of the complete-codeword rules have right-hand sides of length two, and the remaining $r-t$ have length at least three. We show that each of these $t$ rules corresponds to an additional auxiliary rule whose length offsets the saving.

For a two-symbol rule $P_i\to A_iB_i$, exactly one of the two children has an expansion containing $\star$. This expansion contains the complete half-codeword on one side of the star and a strictly shorter piece on the other side; hence it is neither a single terminal (so the child is an auxiliary nonterminal) nor a complete codeword. Associate $P_i$ with this child.

This association is injective: from the longer side of the star in the child's expansion one recovers $u_i$ (reversing it if necessary), and the position of the longer side is determined by the expansion itself, so no child is associated with two different complete codewords. We thus obtain at least $t$ additional auxiliary rules, each of length at least two. The total cost is at least
\[
3(r-t)+2t+2t=3r+t\ge3r.\qedhere
\]
\end{proof}

Hence, an optimal grammar that defines these complete codewords pays at least $3r$ for their rules, regardless of how the rules are arranged. In the lower bound proof below, this cost is subtracted from the size of the target grammar.

\subsection{Construction of the codewords}\label{sec:codewords}

The target alphabet consists of three pairwise disjoint groups of symbols:
\begin{itemize}
\item vertex digits $x_1,\ldots,x_b$ and separator digits $d_1,\ldots,d_b$;
\item vertex labels $a_1,\ldots,a_8$ and separator labels $e_1,\ldots,e_8$;
\item special symbols $\star,\#,\mathsf c_1,\mathsf c_2,\$_1,\ldots,\$_6$.
\end{itemize}
This gives $2b+26$ symbols in total, which is the value of $\sigma$ in~\eqref{eq:TGb}; the parameter $b$ controls the alphabet size.

Codewords use the bijective base-$b$ representation with the least significant digit first, in which every positive integer $i$ is written uniquely as
\[
i=\sum_{j=0}^{\ell-1}r_jb^j,\qquad 1\le r_j\le b.
\]
Let $f_b(i)=x_{r_0}\cdots x_{r_{\ell-1}}$ and $h_b(i)=d_{r_0}\cdots d_{r_{\ell-1}}$. For $1\le i\le2N$ define the inner codewords
\[
V_i^0=f_b(i)\star f_b(i)^R,\qquad
D_i^0=h_b(i)\star h_b(i)^R,
\]
and for $1\le c\le8$ define
\begin{equation}\label{eq:labeled}
V_i^c=a_cV_i^0a_c,
\qquad
D_i^c=e_cD_i^0e_c.
\end{equation}
Each class has $Q=9\cdot2N=18N$ codewords, for a total of $2Q=36N$.

The codewords are nested, so each can be defined by a short rule: removing the outermost pair of digits from an inner codeword leaves an inner codeword of smaller index or a single $\star$, and removing the outer labels from a labeled codeword leaves the corresponding inner codeword. Hence every codeword can be defined by a rule with a right-hand side of length three, and all dependencies point to shorter strings. The total length of all codeword rules is
\begin{equation}\label{eq:codecost}
3\cdot2Q=6Q=108N.
\end{equation}
Although the codeword lengths depend on $b$, every rule has exactly three symbols on its right-hand side, so~\eqref{eq:codecost} is independent of $b$. This is why~\eqref{eq:gTGb} does not depend on $b$.

Next, we fix the order in which codewords occur. The codewords of each class are arranged cyclically in the order
\begin{equation}\label{eq:order}
(V_1^1,\ldots,V_1^8,V_1^0,
V_2^1,\ldots,V_2^8,V_2^0,\ldots),
\end{equation}
and likewise for the class $D$: for each index $i$, the eight codewords with distinct outer labels are listed first, followed by the inner codeword. With this order, any seven consecutive codewords in the cyclic list have pairwise distinct first symbols and, since codewords are palindromes, pairwise distinct last symbols. Section~\ref{sec:repeats} uses this property to rule out repetitions across codewords.

\subsection{Construction of the source string}

Two coloring steps prepare the construction of the source string. Both serve to distinguish adjacent positions by their labels, which restricts repeated factors to a few specific forms.

The first coloring distinguishes the neighbors of each vertex. Join two distinct vertices by an edge whenever they have a common neighbor; the resulting graph has maximum degree at most six and can be colored greedily with seven colors. Denote the color of vertex $i$ by $c_v(i)\in\{1,\ldots,7\}$. Then distinct neighbors of a vertex in the original graph receive distinct colors.

The second coloring distinguishes the junctions in the edge gadget. Order the edges of the original graph arbitrarily, and orient them one by one as $(u_j,v_j)$ so that $v_j\ne u_{j+1}$; since every edge has two distinct endpoints, such an orientation always exists. Build the bridge multigraph with edges $(v_j,u_{j+1})$. Each edge of the original graph contributes at most one bridge incidence to each of its two endpoints, so this graph has maximum degree at most three, and the orientation rules out self-loops. A greedy coloring therefore uses at most five colors; denote the color of the $j$-th bridge edge by $c_e(j)\in\{1,\ldots,5\}$. Labels $6,7,8$ are reserved for the vertex gadgets.

Introduce source symbols $v_i^c,d_i^c$ for the codewords $V_i^c,D_i^c$, and keep the symbols $\#$, $\mathsf c_1$, $\mathsf c_2$, $\$_1,\ldots,\$_6$. Order the source symbols of the two classes according to~\eqref{eq:order} as
\[
\boldsymbol v_1,\ldots,\boldsymbol v_Q,\qquad
\boldsymbol d_1,\ldots,\boldsymbol d_Q.
\]
Let $[j]_Q=1+((j-1)\bmod Q)$ and abbreviate $w_i=v_i^{c_v(i)}$. The source string begins with the prefix
\begin{equation}\label{eq:prefix}
P=\prod_{t=0}^{6}\prod_{j=1}^{Q}
\bigl(\boldsymbol d_j\boldsymbol v_{[j+t]_Q}\bigr).
\end{equation}
The prefix is designed with two goals. The source symbol of every codeword occurs exactly seven times in $P$, which is often enough to force an optimal grammar to create a rule for it; and because of the seven different cyclic shifts, these seven occurrences have pairwise different neighboring codewords, so no fragment spanning several codewords can be reused. Section~\ref{sec:coderules} relies on both properties.

The remainder of the source string encodes the graph. Every vertex occurs once in each of the following five gadgets:
\[
B_1=\prod_{i=1}^N(\#w_i\mathsf c_1d_i^6),\qquad
B_2=\prod_{i=1}^N(\#w_i\mathsf c_2d_i^7),
\]
\[
B_3=\prod_{i=1}^N(w_i\#d_i^7\mathsf c_1),\qquad
B_4=\prod_{i=1}^N(w_i\#d_i^6\mathsf c_2),
\]
\begin{equation}\label{eq:vertexgadgets}
B_5=\prod_{i=1}^N(\#w_i\#d_i^8),
\end{equation}
together with the edge gadget
\begin{equation}\label{eq:edgegadget}
E_G=\left(\prod_{j=1}^{M-1}
\#w_{u_j}\#w_{v_j}\#d_j^{c_e(j)}\right)
\#w_{u_M}\#w_{v_M}\#.
\end{equation}
The source string is
\begin{equation}\label{eq:WG}
W_G=P\$_1B_1\$_2B_2\$_3B_3\$_4B_4\$_5B_5\$_6E_G.
\end{equation}
Since $M\le3N/2<2N$, all indices lie within the range of the code table. The prefix has length $14Q$, the five vertex gadgets have total length $20N$, and the edge gadget has length $6M-1$; adding the six separators gives
\begin{equation}\label{eq:WGlen}
|W_G|=14Q+20N+6M+5=272N+6M+5=\Theta(N).
\end{equation}
Define $C_b(v_i^c)=V_i^c$ and $C_b(d_i^c)=D_i^c$, and leave all special symbols unchanged. The final instance is
\[
T_{G,b}=C_b(W_G).
\]
Note that $W_G$ does not depend on $b$: changing $b$ affects only the alphabet and the codeword lengths of the target string, whereas the source string and its optimum $g(W_G)$ remain unchanged. We can therefore compute $g(W_G)$ once and treat the influence of $b$ separately.

\subsection{Repeated factors and the optimum of the source string}\label{sec:repeats}

Owing to the colorings and the ordering, the repeated factors of the source string can be listed explicitly.

\begin{lemma}\label{lem:repeats}
Every factor of $W_G$ of length at least two that occurs more than once has one of the following three forms:
\begin{equation}\label{eq:threeforms}
\#w_i,\qquad w_i\#,\qquad \#w_i\#.
\end{equation}
\end{lemma}

\begin{proof}
We first consider factors of length two, that is, pairs of adjacent source symbols. In the prefix $P$, apart from the position adjacent to $\$_1$, adjacent pairs alternate between the forms $(\boldsymbol d_j,\boldsymbol v_{[j+t]_Q})$ and $(\boldsymbol v_{[j+t]_Q},\boldsymbol d_{j+1})$, with indices taken modulo $Q$. For fixed $j$, the seven values of $t$ give seven different vertex symbols. Hence no adjacent pair of the prefix repeats. In the suffix, two codeword symbols are always separated by $\#$, $\mathsf c_1$, or $\mathsf c_2$, so adjacent pairs of the prefix do not reappear in the suffix either.

In the suffix, each $\$_j$ occurs only once. The symbols $d_i^6$ occur in $B_1$ and $B_4$: their left neighbors are $\mathsf c_1$ and $\#$, and their right neighbors lie in $\{\#,\$_2\}$ and in $\{\mathsf c_2\}$, respectively. The symbols $d_i^7$ occur in $B_2$ and $B_3$: their left neighbors are $\mathsf c_2$ and $\#$, and their right neighbors lie in $\{\#,\$_3\}$ and in $\{\mathsf c_1\}$, respectively. So no adjacent pair on either side of these symbols repeats. The symbol $d_i^8$ occurs only once, in $B_5$, and each $d_j^{c_e(j)}$ in the edge gadget is determined by the index $j$. The vertex symbols next to $\mathsf c_1,\mathsf c_2$ are distinguished by the index $i$. The only adjacent pairs that can repeat are therefore $\#w_i$ and $w_i\#$.

Next, consider longer factors. Every adjacent pair in a repeated factor of length at least three must itself repeat, so the factor is built from the two kinds of pairs above; that is, it contains $\#w_i\#$ or $w_i\#w_j$. The latter form occurs only in the edge gadget, and since every edge is written only once, the same ordered pair $(i,j)$ does not repeat.

Finally, we show that no extension of $\#w_i\#$ repeats. In $B_5$, its left neighbor is $d_{i-1}^8$ or $\$_5$ and its right neighbor is $d_i^8$. In the edge gadget, if $w_i=w_{u_j}$, its left neighbor is $d_{j-1}^{c_e(j-1)}$ or $\$_6$ and its right neighbor is $w_{v_j}$; if $w_i=w_{v_j}$, its left neighbor is $w_{u_j}$ and its right neighbor is $d_j^{c_e(j)}$ or the end of the string. Separators of different edges are distinguished by their indices, and different edges incident to $i$ have different other endpoints. Hence no extension on either side repeats, and no factor of length at least four repeats either.
\end{proof}

By Lemma~\ref{lem:repeats}, every auxiliary expansion of an optimal grammar has one of the three forms in~\eqref{eq:threeforms}, so $g(W_G)$ can be computed directly. We first add all rules of length two:
\[
L_i\to\#w_i,\qquad R_i\to w_i\#.
\]
Adding these rules does not increase the optimum: defining $L_i$ costs two symbols, and $B_1$ and $B_2$ each contain an occurrence of $\#w_i$ that can be replaced by $L_i$, saving exactly two symbols; the same holds for $R_i$ in $B_3$ and $B_4$. The only other potentially useful rules are $A_i\to L_i\#$, with expansion $\#w_i\#$. Let $C$ be the set of vertices $i$ for which $A_i$ is used; as we will see, $C$ plays the role of a vertex cover.

We now compute shortest parsings gadget by gadget. The prefix contains no repeated factors and must be written with $14Q$ symbols. Defining all $L_i,R_i$ costs $4N$ symbols, and defining the $A_i$ another $2|C|$. In $B_1$ through $B_4$ each vertex needs three symbols; in $B_5$ a vertex in $C$ needs two and any other vertex three.

The dependence on $C$ arises in the edge gadget. The fragment of the edge $\{u_j,v_j\}$ is~$\#w_{u_j}\#w_{v_j}\#$. If at least one endpoint lies in $C$, it can be written with two symbols, using one $A$ rule and one $L$ or $R$ rule; if neither endpoint lies in $C$, three symbols are needed. In addition, each edge except the last requires one symbol $d_j^{c_e(j)}$. Since Lemma~\ref{lem:repeats} lists all repeated factors, the local costs are independent and can be added. Including the six symbols $\$_j$, the total cost is
\[
\begin{split}
&14Q+4N+2|C|+12N+(3N-|C|)
+(3M-1+|\{\{u,v\}\in E(G):u,v\notin C\}|)+6\\
&\quad=14Q+19N+3M+5+|C|
+|\{\{u,v\}\in E(G):u,v\notin C\}|.
\end{split}
\]
It remains to minimize the last two terms, $|C|+|\{\{u,v\}\in E(G):u,v\notin C\}|$, whose minimum is exactly $\tau(G)$: for every $C$, adding one endpoint of each uncovered edge yields a vertex cover of size at most this sum, and taking $C$ to be a minimum vertex cover makes the second term vanish. Therefore
\begin{equation}\label{eq:gWG}
\boxed{g(W_G)=14Q+19N+3M+5+\tau(G)
=271N+3M+5+\tau(G).}
\end{equation}

This completes the analysis of the source string. We now designate anchors in the target string $T_{G,b}$: the central $\star$ of every codeword carries the label of the corresponding source symbol, and every unencoded special symbol serves as its own anchor and label. Reading these labels in order recovers exactly $W_G$.

To apply the Projection Theorem, we must also analyze the repeated factors of the target string.

\begin{lemma}\label{lem:targetrepeats}
Every repeated factor of $T_{G,b}$ that contains $\star$ is contained in a single codeword or in a single $\#V_i^{c_v(i)}\#$. In particular, it cannot span two codeword centers.
\end{lemma}

\begin{proof}
The proof has three steps: we show that alignments across codeword boundaries are rigid, exclude factors with two centers, and bound how far a factor with one center can extend.

First, consider a fragment containing a codeword center and crossing a boundary of that codeword. We show that every other occurrence crosses the same side of the same codeword. The symbol sets are disjoint: apart from the central $\star$, vertex codewords use only vertex digits and vertex labels, and separator codewords use only separator digits and separator labels, and these two groups are disjoint. Every codeword has the form $u\star u^R$, and $u$ determines it uniquely.

Suppose a fragment crosses the right boundary of a codeword. Then after the $\star$ it contains the complete $u^R$ and also the first symbol outside the codeword. In the prefix, this outside symbol belongs to the other class of codewords; in the suffix, it is $\#$, $\mathsf c_1$, $\mathsf c_2$, or $\$_j$. In both cases it cannot coincide with a symbol inside the codeword. Hence, in any other occurrence the codeword must have the same length, since otherwise some position would match an outside symbol with an inside symbol. Equal right halves in turn determine the whole codeword. The left boundary is handled symmetrically. Thus, whenever a repeated fragment crosses a codeword boundary, both occurrences cross the same side of the same codeword.

Second, we exclude factors with two centers. If a repeated factor contains two occurrences of $\star$, take two adjacent ones. The fragment from the first center to the second crosses the right boundary of the first codeword and the left boundary of the second; by the first step, both codewords and all special symbols between them are uniquely determined. Hence the corresponding factor of the source string $W_G$ also repeats and contains two encoded source symbols, contradicting Lemma~\ref{lem:repeats}. So a repeated factor contains at most one $\star$.

Third, we bound how far a fragment with one center can extend. In what follows, the prefix and the gadgets refer to their encoded versions. The same codeword occurs seven times in the prefix, and~\eqref{eq:order} ensures that in every occurrence the first and last symbols of the neighboring codewords differ. Moreover, the neighbors of a codeword in the suffix are special symbols, while in the prefix they belong to the other class of codewords. Consequently, a repeated fragment in the prefix can neither cross a codeword boundary nor have a boundary-crossing occurrence in the suffix.

In the suffix, the separator codewords $D_i^6$ occur in $B_1$ and $B_4$ with left neighbors $\mathsf c_1$ and $\#$, and right neighbors in $\{\#,\$_2\}$ and in $\{\mathsf c_2\}$, respectively. The codewords $D_i^7$ occur in $B_2$ and $B_3$ with left neighbors $\mathsf c_2$ and $\#$, and right neighbors in $\{\#,\$_3\}$ and in $\{\mathsf c_1\}$, respectively. Both $D_i^8$ and each $D_j^{c_e(j)}$ occur only once in the suffix. Hence a repeated factor containing the center of a separator codeword cannot cross a codeword boundary.

Write $\widehat w_i=C_b(w_i)=V_i^{c_v(i)}$ for the codeword of the source symbol $w_i$. These are exactly the vertex codewords occurring in the suffix. In $B_1$ and $B_2$ its right neighbor is $\mathsf c_1$ and $\mathsf c_2$, respectively; in $B_3$ and $B_4$ its left neighbor is $\mathsf c_1$ and $\mathsf c_2$, respectively, or a gadget separator occurring only once. The only other outside symbol that repeats is $\#$. Table~\ref{tab:beyondhash} lists the first symbol seen beyond this $\#$; when the end of the string or a gadget separator is reached, the occurrence itself is unique.

\begin{table}[ht]
\centering
\small
\begin{tabular}{@{}lll@{}}
\toprule
Occurrence & Left of $\#\widehat w_i$ & Right of $\widehat w_i\#$ \\
\midrule
$B_1,B_2,B_5$ & $e_6,e_7,e_8$, respectively & $e_8$ (in $B_5$) \\
$B_3,B_4$ & --- & $e_7,e_6$, respectively \\
edge gadget, $i=u_j$ & $e_{c_e(j-1)}$ & $a_{c_v(v_j)}$ \\
edge gadget, $i=v_j$ & $a_{c_v(u_j)}$ & $e_{c_e(j)}$ \\
\bottomrule
\end{tabular}
\caption{The first symbol beyond the $\#$ adjacent to $\widehat w_i$.}
\label{tab:beyondhash}
\end{table}

The edge colors in the table occur only if the corresponding bridge edge exists. For fixed $i$, different neighbors have different colors $c_v$, different bridge edges incident to $i$ have different colors $c_e$, and the vertex labels $a_c$ and separator labels $e_c$ come from different groups of symbols. Hence the symbols to be compared on the same side of the table never coincide. Therefore, a repeated fragment can extend beyond $\widehat w_i$ by at most the adjacent $\#$, since one step further the symbols differ. This proves the lemma.
\end{proof}

\begin{corollary}\label{cor:projcond}
If some occurrence of a factor of $T_{G,b}$ contains at least two anchors, then all its occurrences have the same projection onto source labels.
\end{corollary}

\begin{proof}
We distinguish cases according to the number of codeword centers in the factor. A factor containing two centers cannot repeat by Lemma~\ref{lem:targetrepeats}. A repeated factor containing exactly one center lies within a single codeword or within a single $\#\widehat w_i\#$; the proof of Lemma~\ref{lem:targetrepeats} also shows that codewords and boundaries have the same relative positions in all occurrences, so the labels of all anchors agree.

It remains to consider factors containing no codeword center but at least two anchors. These anchors can only be unencoded special symbols. If a codeword lies between two special symbols, its central $\star$ also falls into the factor, which is the previous case; if no codeword lies between them, the consecutive special symbols contain some $\$_j$ occurring only once, so the factor does not repeat.
\end{proof}

Thus, all conditions of the Projection Theorem are satisfied. When a target grammar is translated back to $W_G$, the translation is consistent and, by~\eqref{eq:projthm2}, we may subtract the rule lengths of all auxiliary nonterminals each of whose occurrences covers at most one anchor; codeword rules are of this kind.

\subsection{Codeword rules in optimal grammars}\label{sec:coderules}

\begin{lemma}\label{lem:allcodes}
Every reduced optimal grammar for $T_{G,b}$ contains auxiliary nonterminals for all $2Q$ complete codewords.
\end{lemma}

\begin{proof}
We argue by contradiction, charging the cost to the prefix. Take a reduced optimal grammar and suppose that $s>0$ complete codewords are missing.

The symbol $\$_1$ occurs only once in $T_{G,b}$, so an auxiliary rule containing $\$_1$ cannot be referenced twice, and $\$_1$ must be written directly by the start rule. Hence the symbols of the start rule before $\$_1$ generate exactly $C_b(P)$.

This prefix contains $14Q$ codeword centers. By Lemma~\ref{lem:targetrepeats}, each right-hand side symbol in the prefix that contains a center generates at most one center: if it generated two, its expansion would have to repeat, which Lemma~\ref{lem:targetrepeats} excludes. Hence the prefix needs at least $14Q$ right-hand side symbols containing centers.

We now show that missing codewords incur an extra cost. For each missing codeword, none of its seven designated occurrences in the prefix can be generated by exactly one right-hand side symbol, since that symbol would be an auxiliary nonterminal expanding to this complete codeword, contradicting that it is missing. Symbols containing a center cannot cross codeword boundaries of the prefix, so at least part of each designated occurrence must be generated by symbols containing no center. A symbol containing no center touches at most two adjacent designated codeword occurrences, since otherwise it would cover the center of the codeword in between. Hence, besides the $14Q$ symbols above, at least $\lceil7s/2\rceil$ symbols without a center are needed. In total, the start rule uses at least
\[
14Q+\left\lceil\frac{7s}{2}\right\rceil
\]
symbols on $C_b(P)$.

We now use an exchange argument and add rules for all missing codewords following the construction of Section~\ref{sec:codewords}. Every new rule has a right-hand side of length three and depends only on shorter codewords, and no existing rule depends on a new nonterminal, so no cycle is created. After this, the part of the start rule covering $C_b(P)$ can be rewritten using complete-codeword nonterminals with exactly $14Q$ references; all rules outside the prefix and the rest of the start rule remain unchanged. The net change of the grammar size is at most
\[
3s-\left\lceil\frac{7s}{2}\right\rceil<0,
\]
which is negative since $7/2>3$ (this is why the prefix contains seven occurrences of each codeword), contradicting optimality. Hence $s=0$.
\end{proof}

We can now combine these results. Let $H$ be a reduced optimal grammar for $T_{G,b}$. By Lemma~\ref{lem:allcodes}, it contains $r=2Q$ complete-codeword rules; by Lemma~\ref{lem:codecost}, these rules together with the related single-center auxiliary rules cost at least $3r=6Q$. The expansion of each of these rules contains exactly one $\star$ and no other anchor, so each occurrence covers at most one anchor, and their lengths may be subtracted in~\eqref{eq:projthm2}. Corollary~\ref{cor:projcond} and~\eqref{eq:projthm2} give
\[
g(W_G)\le |H^\downarrow|
\le |H|-6Q
=g(T_{G,b})-6Q.
\]
Conversely, take an optimal grammar for $W_G$, replace every encoded source symbol by the corresponding codeword nonterminal, and add all codeword rules of Section~\ref{sec:codewords}, of total length exactly $6Q$. The old rules keep their lengths, so
\[
g(T_{G,b})\le g(W_G)+6Q.
\]
The two bounds coincide. In contrast to Section~\ref{sec:npc}, the encoding cost is thus computed exactly rather than bounded by a constant. Therefore
\begin{equation}\label{eq:gTGbexact}
\boxed{g(T_{G,b})=g(W_G)+6Q=g(W_G)+108N.}
\end{equation}
Substituting~\eqref{eq:gWG} yields~\eqref{eq:gTGb}.

\subsection{Size of the reduction and quantitative consequences}

Finally, we bound the size of the reduction with constants uniform in the alphabet parameter. Let $\ell_b(i)$ be the length of the bijective base-$b$ representation of $i$. Inner codewords have length $2\ell_b(i)+1$, labeled codewords $2\ell_b(i)+3$, and the nine codewords of the same index have total length $18\ell_b(i)+25$. The prefix contains every codeword of both classes seven times, so
\[
|C_b(P)|=14\sum_{i=1}^{2N}(18\ell_b(i)+25).
\]
Uniformly for all admissible $b$,
\[
\sum_{i=1}^{2N}\ell_b(i)
=\Theta\!\bigl(N(1+\log_bN)\bigr).
\]
The suffix contains $O(N+M)=O(N)$ codewords, each of length $O(1+\log_bN)$, hence
\begin{equation}\label{eq:TGblen}
|T_{G,b}|=\Theta\!\bigl(N(1+\log_bN)\bigr).
\end{equation}
Since $b\le2N$, all digits, outer labels, and special symbols actually occur in the string, so the alphabet size is exactly $\sigma=2b+26$.

Let $n=|T_{G,b}|$. By~\eqref{eq:TGblen} and $2\le b\le2N$,
\[
\log n=\Theta(\log N),\qquad
\log\sigma=\Theta(\log b).
\]
Consequently,
\begin{equation}\label{eq:nR}
\frac{n\log\sigma}{\log n}
=\Theta\!\left(
N(1+\log_bN)\cdot\frac{\log b}{\log N}
\right)
=\Theta(N),
\end{equation}
where we used $\log b+\log N=\Theta(\log N)$. By~\eqref{eq:nR}, $n/R=\Theta(N)$ for every choice of $b$. Moreover, by~\eqref{eq:gTGb}, $\tau(G)$ can be read off directly from $g(T_{G,b})$. Thus any running-time bound for SGP at these parameters $(n,\sigma)$ yields a running-time bound for Vertex Cover. Section~\ref{sec:transfer} derives the resulting time bounds for SAT; no assumption such as ETH is involved.

Choosing $b$ as a constant, as $(\log N)^a$, as $2^{(\log N)^\beta}$, or as a fixed positive power of $N$ produces the alphabet regimes considered in Table~\ref{tab:regimes}. These constructions determine the reduction lengths; the running times in that table come from the algorithms. For smaller fixed alphabets, use Section~\ref{sec:npc}.

\section{Quantitative transfer from grammar algorithms to SAT}\label{sec:transfer}

The reductions of Sections~\ref{sec:npc} and~\ref{sec:eth} have quantitative
consequences that do not depend on ETH. They map SAT instances to SGP
instances, so they turn an improved algorithm for SGP into an improved
algorithm for SAT. They do not turn a SAT algorithm into an SGP algorithm;
that requires a separate construction, such as the one in
Section~\ref{sec:sparse-interval-tensor}.

Let $\mathcal T(x,s)$ be a nondecreasing upper bound on the worst-case running
time of a fixed SGP decision algorithm on all strings of length at most $x$
and alphabet size at most $s$. Any upper bound can be made nondecreasing by taking
its maximum over smaller parameter values. This convention lets us substitute
bounds of the form $\Theta(x)$ without assuming further regularity.

\begin{theorem}[Unconditional quantitative transfer]\label{thm:reverse-transfer}
Let $F$ be a $3$-CNF with $m\ge1$ clauses, each containing three distinct
variables. Given any integer $2\le b\le18m$, one can construct in time
polynomial in $|F|$ an SGP instance with
\[
 s=2b+26,\qquad n\le C m(1+\log_b(m+1)),
\]
where $C$ is absolute and independent of $b$. Hence $F$ is decidable in
\[
 \poly(|F|)+\mathcal T\!\left(
 \left\lceil C m(1+\log_b(m+1))\right\rceil,2b+26\right).
\]
There is also a binary reduction with $n\le C'm\log(m+2)$, giving time
\[
 \poly(|F|)+\mathcal T\!\left(
 \left\lceil C'm\log(m+2)\right\rceil,2\right).
\]
The constant $C'$ is absolute. Neither assertion assumes ETH.
\end{theorem}
\begin{proof}
We follow the construction of Appendix~\ref{app:eth}. For a variable with $d\ge2$ occurrences,
make a cycle on $2d$ vertices whose alternating classes represent its two values;
for $d=1$, use an edge. Make one triangle per clause, and join each literal
vertex to the vertex representing the falsifying value at its own occurrence.
The resulting graph $G$ has maximum degree three and
\[
 v=9m,\qquad e=12m-u_1,
\]
where $u_1$ counts variables with one occurrence. Independent sets have size at
most $4m$: the variable gadgets contribute at most $3m$ and the triangles at
most $m$. Equality requires a complete alternating class in every variable
gadget and a true literal in every clause, and is possible exactly when $F$ is
satisfiable. Thus $F$ is satisfiable if and only if $\tau(G)\le5m$.

Section~\ref{sec:eth} gives, for $2\le b\le2v$, a string with alphabet
$2b+26$, length $\Theta(v(1+\log_b v))$ with uniform absolute constants, and
\[
 g(T_{G,b})=379v+3e+5+\tau(G).
\]
We use the threshold $379v+3e+5+5m$. Substituting $v=9m$ proves the first assertion.

For the binary version, Section~\ref{sec:npc} uses $D=3072$ and
$h=D+1=3073$ copies of $G$. The encoded binary string satisfies
\[
 B+h\tau(G)\le g(T)\le B+h\tau(G)+D,
 \qquad B=299hv+3he+5.
\]
Therefore $g(T)\le B+5hm+D$ if and only if $\tau(G)\le5m$: in a no-instance
the lower bound exceeds this threshold by at least $h-D=1$.
Appendix~\ref{app:fixed} gives length $\Theta(v\log(v+1))$, since the copy count and
codeword length are constant. The binary encoding preserves the decision gap,
not an exact formula for the target optimum.
\end{proof}

Clauses with fewer than three literals increase the size by only a constant factor. Replace $a\vee b$ by
$(a\vee b\vee z)\wedge(a\vee b\vee\neg z)$ with a fresh variable, and a
unit clause $a$ by the four clauses $a\vee z^{\epsilon}\vee w^{\eta}$ for all
choices of signs. First remove tautologies and repeated literals, and handle an
empty clause directly. These transformations preserve satisfiability, not
necessarily counts. Starting with $N$ variables and $m$ clauses, normalization
produces $O(N+m)$ variables and $O(m)$ clauses. Hence a fixed-density input
with $m=O(N)$ remains of linear size in the original variable bound $N$,
even if deleting unused variables reduces the number of occurring variables.

\begin{corollary}[A logarithmic speed scale]\label{cor:log-speed-scale}
Fix $\delta>0$ and $a\in\mathbb R$. If binary SGP has a uniform algorithm in
\begin{equation}\label{eq:binary-hypothesis}
 \poly(n)\,2^{O\left(
 n(\log\log n)^a/(\log n)^{1+\delta}\right)},
\end{equation}
then, for every fixed clause density, $3$-SAT on $N$ variables is decidable in
\begin{equation}\label{eq:sparse-transfer}
 \poly(N)\,2^{O\left(
 N(\log\log N)^a/(\log N)^\delta\right)}.
\end{equation}
Constants in the conclusion may depend on the fixed density.
\end{corollary}
\begin{proof}
The binary reduction gives $n=O(N\log N)$. The exponent in
\eqref{eq:binary-hypothesis} is eventually increasing, while
$\log(N\log N)=\Theta(\log N)$ and
$\log\log(N\log N)=\Theta(\log\log N)$. Substitution proves the claim.
\end{proof}

For the tunable instances, put $R=\log_s n$. Uniformly in $2\le b\le18m$,
the calculations of Section~\ref{sec:eth} give $n/R=\Theta(m)$ and
$\log n=\Theta(\log(m+1))$. Thus a uniform alphabet-sensitive bound
\[
 \poly(n)\,2^{O\left(
 \frac{n}{R}\frac{(\log\log n)^a}{(\log n)^\delta}\right)}
\]
also implies~\eqref{eq:sparse-transfer}. If an algorithm is guaranteed only on
an alphabet range, that range must contain the actual output pair $(n,2b+26)$;
a bound proved for one exact alphabet size does not automatically apply to a nearby size.

For example, $a=0$ and $\delta=1/2$ in the binary hypothesis transfer to the
$2^{O(N/\sqrt{\log N})}$ fixed-density SAT bound of~\cite{sat}.
Any fixed $\delta>1/2$ would improve the power of $\log N$ in that exponent.
This is an implication; it does not rule out such an algorithm.

\subsection{The loss from sparsifying a general formula}

The fixed-density conclusion does not extend to growing densities, since
its constants depend on the density. Instead, we use the quantitative sparsification
lemma recorded in~\cite[Lemma labelled \texttt{lem:sparsify}]{sat}. At width
three and supplied rational $0<\varepsilon\le1$, it produces a tree with at most
\[
 \poly(N,\varepsilon^{-1})2^{\varepsilon N}
\]
nodes, whose leaves have at most $B_\varepsilon N$ clauses, where
\[
 B_\varepsilon=O\!\left(
 \varepsilon^{-2}\log^2(2/\varepsilon)\right).
\]
Leaves with no clauses or with an empty clause are decided directly; apply
the binary reduction to every other leaf. For an input with $N$ variables and
$M$ clauses the resulting unconditional transfer bound is
\begin{equation}\label{eq:dense-transfer-function}
 \poly(N,M,\varepsilon^{-1})2^{\varepsilon N}
 \left[1+\mathcal T\!\left(
 \left\lceil C''B_\varepsilon N\log(B_\varepsilon N+2)\right\rceil,2
 \right)\right],
\end{equation}
when the parameter has polynomial-size encoding. The absolute constant $C''$
includes short-clause normalization. Here the SGP algorithm is applied
directly, and no constant depends on a fixed density.

\begin{corollary}[An explicit general-density consequence]\label{cor:dense-transfer}
For fixed rational $\delta>0$ and $a$, hypothesis~\eqref{eq:binary-hypothesis}
implies an algorithm for arbitrary $3$-SAT in
\[
 \poly(N,M)\,
 2^{O\left(N(\log\log N)^{(a+2)/3}/(\log N)^{\delta/3}\right)}.
\]
This bound is what the sparsification route yields; we do not claim that
the loss of a factor of three in the logarithmic exponent is necessary.
\end{corollary}
\begin{proof}
Write $t=\log N$, $u=\log\log N$, and choose a positive rational
$\varepsilon=\Theta(t^{-\delta/3}u^{(a+2)/3})$ with polynomial-size encoding.
For large $N$ it lies in $(0,1]$ and its reciprocal is polylogarithmic.
Since $\log(2/\varepsilon)=\Theta(u)$,
\[
 B_\varepsilon=O\!\left(t^{2\delta/3}u^{(2-2a)/3}\right).
\]
The SGP exponent in~\eqref{eq:dense-transfer-function} is at most
\[
 O(B_\varepsilon N u^a/t^\delta)
 =O(Nt^{-\delta/3}u^{(a+2)/3}),
\]
using $\log(B_\varepsilon N\log(B_\varepsilon N))=\Theta(t)$.
The sparsification exponent $\varepsilon N$ has the same order, and all
remaining factors are polynomial.
\end{proof}

\section{Discussion}\label{sec:discussion}

Theorem~\ref{thm:subexp-intro} improves on part~2 of Theorem~\ref{thm:main} by a factor of $\Theta\bigl(\log(R+1)\sqrt{\log n}/\log\log n\bigr)$ in the exponent; for fixed alphabets this factor is $\Theta(\sqrt{\log n})$. This improvement does not come from a better choice of the length threshold of Section~\ref{sec:alphalg}: it comes from the grammar-size budget and the sparsity of the selected interval choices.

The parameterized bound in Theorem~\ref{thm:sparse-budget} suggests several
directions for improvement. The current subset indices treat the active
wires as an arbitrary sparse subset of $O(nL)$ candidates. Exploiting the
nonoverlap and coverage structure of the active wires may shorten these
indices. Any such improvement must also account for the precision needed to
distinguish a positive certificate count from zero. Two further questions are
whether the running time can be expressed in terms of the actual optimum
$g(T)$, and whether comparable times can be achieved in linear space.

The reductions of Section~\ref{sec:transfer} serve as quantitative benchmarks.
On the tunable instances produced from a formula with $m$ clauses,
$n/R=\Theta(m)$ and $\log n=\Theta(\log(m+1))$. For $m=O(N)$,
the present algorithm therefore transfers to sparse
$3$-SAT in $2^{O(N\log\log N/\sqrt{\log N})}$ time up to polynomial factors.
This does not improve the fixed-density SAT bound of~\cite{sat}. A binary
SGP algorithm with exponent
$O(n(\log\log n)^a/(\log n)^{1+\delta})$ would transfer to sparse SAT
with exponent $O(N(\log\log N)^a/(\log N)^\delta)$; for fixed $a$ and
$\delta>1/2$, it would improve that logarithmic power. General formulas
require the separately quantified sparsification loss. None of these
implications assumes ETH.

\section{Conclusion}\label{sec:conclusion}

Every string has an optimal grammar whose dictionary is given by pairwise disjoint occurrences in the input. This fact yields self-contained linear-space algorithms, including one in $O^*(1.872258^n)$ time, and a sparse interval certificate. Combined with the Boolean tensor contraction
theorem of~\cite{sat}, the certificate gives a deterministic polynomial-space exact
algorithm with exponent
$O(n\log\sigma\,\log\log n/(\log n)^{3/2})$, uniformly in the alphabet size.
Anchor-based projections establish binary NP-completeness and give
reductions with tunable alphabet size. Without any assumption such as ETH,
these reductions quantify how further improvements for SGP would accelerate
SAT, separately for fixed-density and general formulas.

\appendix

\section{Reduction lengths and algorithm transfers}\label{app:eth}

\subsection{General alphabets}

\begin{proposition}\label{prop:general}
There is a polynomial-time reduction from $3$-SAT with $m$ clauses to general-alphabet SGP with output length $O(m)$. In particular, a general-alphabet SGP running-time bound transfers to sparse $3$-SAT under a linear change of the length parameter.
\end{proposition}

\begin{proof}
First assume that every clause contains literals of exactly three distinct variables. The normalizations in Section~\ref{sec:transfer} reduce other $3$-CNFs to this case with at most a constant-factor increase in the clause count; empty clauses and clause-free formulas are decided directly. Thus it suffices to prove a length bound linear in the normalized clause count $m$. For a fixed-density input with $N$ declared variables, this count is still $O(N)$, whether or not some variables are unused.

We first construct an Independent Set instance of maximum degree at most three. For a variable occurring $d\ge2$ times, build an even cycle of length $2d$ whose vertices alternately represent the values true and false; a variable occurring only once is represented by a single edge. For every clause, build a triangle whose three vertices represent its literals, and connect each of them to the occurrence vertex of the corresponding variable gadget that represents this literal being false. Since different literal occurrences use different vertices, the maximum degree is at most three.

We next bound the size of independent sets. Variable gadgets contribute at most $3m$ vertices and clause triangles at most $m$. The total reaches $4m$ if and only if every variable gadget selects a complete alternating class and every clause selects a true literal, which is exactly equivalent to satisfiability of the formula.

The resulting graph $H=(V,E)$ has $|V|=9m$; if $u_1$ variables occur only once, then $|E|=12m-u_1\ge9m=|V|$. Finally, we apply the reduction of Charikar et al.~\cite[Section~V-A]{CLL05}, which turns this graph into a string $T_H$ of length $O(|V|+|E|)$ with
\[
g(T_H)=15|V|+3|E|+\tau(H).
\]
Since $\tau(H)=|V|-\alpha(H)$, the formula is satisfiable if and only if
\[
g(T_H)\le15|V|+3|E|+5m.
\]
The whole reduction has output length $O(m)$, and hence $O(N)$ for fixed-density inputs. Composing it with any SGP decision algorithm proves the asserted transfer.
\end{proof}

\subsection{Fixed alphabets}\label{app:fixed}

\begin{observation}\label{obs:vc}
A normalized $3$-SAT formula with $m$ clauses, each containing three distinct variables, reduces to Vertex Cover on a graph of maximum degree three with $9m$ vertices. Normalizing shorter clauses increases the clause count by at most a constant factor.
\end{observation}

\begin{proof}
The construction in the proof of Proposition~\ref{prop:general} reduces sparse 3-SAT to Independent Set on graphs of linear size and maximum degree at most three. Since the complement of an independent set is a vertex cover, the threshold carries over directly, and the graph has the stated size.
\end{proof}

The $24$-letter reduction of Casel et al.~\cite{CFGGS21} satisfies~\eqref{eq:casel}, but Corollary~\ref{cor:ethbinary} also requires its output length, which is not stated explicitly. We derive it from the construction in the conference version~\cite[Section~3]{CFGGS16}.

\begin{lemma}\label{lem:caselLen}
This reduction satisfies $|w_G|=\Theta(N\log N)$.
\end{lemma}

\begin{proof}
The construction uses centrally palindromic codewords
\[
\langle i\rangle=\varphi(i)\star\varphi(i)^R,
\]
where $\varphi(i)$ is the bijective base-$7$ representation, so $|\varphi(i)|=\Theta(\log(i+1))$. For the upper bound, all indices are $O(N)$ and the total number of codeword occurrences is $\Theta(N)$, so $|w_G|=O(N\log N)$.

For the lower bound, the prefix of the construction contains $\Theta(N)$ codewords with indices $\Omega(N)$, each of length $\Omega(\log N)$; multiplying gives the matching lower bound.
\end{proof}

Section~\ref{sec:npc} takes only a constant number of copies of the source graph and replaces symbols by codewords of constant length; neither step changes the order of magnitude. The final binary instance therefore satisfies
\[
n=\Theta(N\log N),\qquad n/\log n=\Theta(N),
\]
and together with Observation~\ref{obs:vc} this gives Corollary~\ref{cor:ethbinary}.

\section{An explicit binary code for 24 letters}\label{app:code}

Lemma~\ref{lem:code} proves the existence of synchronizing codes by the probabilistic method; for a deterministic reduction over a fixed alphabet, we list $24$ explicit codewords of length $128$ and verify boundary synchronization (Section~\ref{sec:anchors}) by a program. For an ordered pair of symbols $(a,b)$, the boundary fragment is
\[
F(a,b)=C(a)[97..128]\,C(b)[1..32].
\]
The program exhaustively checks all $24^2=576$ boundary fragments against all $128$ starting phases in the concatenation of every pair of codewords. The output $(\mathrm{True},576)$ certifies that all boundary fragments are pairwise distinct and that no fragment occurs at a wrong phase or for a wrong pair of symbols, which is exactly property (SI). The listing is also provided as the ancillary file \texttt{check\_sync.py}.

\begin{lstlisting}
CODES = [
    "1101111110010010100110111011100010110100000100110110101101101000"
    "0110000001100111110101100010010110000000000101000000100010100011",
    "1001001011100000110100111111001010100111101110001000111011001010"
    "1100111101000011100000000110011001101010010000111011110011000111",
    "1010010011101010010011111001111010110110010000001001111010101010"
    "1000100111000000111001011111100011100101110000000101100111000101",
    "1110000111011111100001100110100110111001000010101100110101100011"
    "1011000110000110011001001001110111010010110001000000011000010111",
    "0001111110001010001000111000100011101110011011011101110001001101"
    "1001101101010110010101101011001001111011100001110000111000111111",
    "0000111011110111101101000111000101000110010011110111010100011111"
    "1001001000010001011110111110110000101000101010100111110001001001",
    "1111110010011101111000000101011000100111101011111000001001010011"
    "0101000110111000110010000010000110100111100011101010111000101101",
    "1010100101000000100110011011011101110011000001000000010100111001"
    "0110101011111011101011011100101101110010100100011111111000010101",
    "1001100000111110111110100011101001001100010101010001011111110101"
    "0101111011011011000110000101011011001100001001010001011111010000",
    "1001010101001000111111001000010010110010101111110011010100010010"
    "1111001100011010010011000000101100011001001110100010000100000010",
    "1110010100110110110101101110101011000101110110011101011011001011"
    "1011011101000100111001000000001111011011011011111110000000101000",
    "0111011110000001011100101010110101011111001010100111011001010000"
    "0001101110000111001110011011001011110011101111111011010100110011",
    "1101110100111010010101000100001000001001000011101010101111001001"
    "0001010100101101101110101011101011010111001111000100010101101100",
    "1010101010100110000001001111001111111000001101001001100011100100"
    "1110000101101101110100011011110010000100111011100100101101010100",
    "0100100001011011001101001110001001011011001101011000101110001011"
    "1011110110001100010001011001111110010000100100010000010100101101",
    "1000000011111110000000101010011111001011111011111111001010011111"
    "1001100001000011000010111101000011010100101000100011100101111110",
    "1001111010000110111010000001110100100010010111010000011110100001"
    "0000101010010111110011111100011000001001111111100111000101100001",
    "0101111100110010111001011110110111001011010100100110010010111011"
    "1000001011111100010111111000001000001010011000000111010001000110",
    "1010000010010100111010011111011001000101100011110111000111010100"
    "1010000001100001000001111011000100000110111110011000100110100110",
    "1001111001010011001100001100110100100001111000010000000001011110"
    "0010111110100000011011011010110110101011011100010110111101100011",
    "0011101100011101111100100011010101011010000111001110110000101010"
    "0100001001101010001010110110101011111111000110111000100110011111",
    "0000100110010110011110001111011101110111111110101111001010100100"
    "0000011011110010001111101110010110101000101100000010000000000110",
    "0110001100000111110001111100001000011011000110010110001101110101"
    "0110101011010011001001101000101000011111001101010011100111010001",
    "0001011111110111111010111001010011111100101100001011101111001110"
    "1000010001100010101101101000111001110011111101010100111111001010",
]

def check_sync(codes):
    m, half, window = 128, 32, 64
    assert len(codes) == 24 == len(set(codes))
    assert all(len(c) == m and set(c) <= {"0", "1"} for c in codes)

    fingerprints = {}
    for a, left in enumerate(codes):
        for b, right in enumerate(codes):
            f = left[-half:] + right[:half]
            if f in fingerprints:
                return False, ("fingerprint collision", (a, b))
            fingerprints[f] = (a, b)

    for c, left in enumerate(codes):
        for d, right in enumerate(codes):
            joined = left + right
            for r in range(m):
                w = joined[r:r + window]
                if w in fingerprints:
                    if r != m - half or fingerprints[w] != (c, d):
                        return False, ("spurious occurrence", (c, d, r))
    return True, len(fingerprints)

print(check_sync(CODES))
\end{lstlisting}

\section{Finite verification of the interval encoding}\label{app:verification}

The accompanying Python script, which uses only the standard library,
checks the local interval network and its subset indices. It does not implement the asymptotic
tensor-contraction theorem. As an independent reference, an
oracle enumerates dictionaries of repeated proper substrings and computes
shortest parsings of right-hand sides by dynamic programming; it uses
neither tensor factors nor sparse-wire indices.

For 332 strings---every binary string of lengths three through seven and
every ternary string of lengths three and four, without duplicates---the
oracle supplies an optimum grammar. The script finds disjoint occurrences
of its dictionary, transfers the parsings to them, and splits the certificate
at several short-interval thresholds. With the resulting valid long patterns
fixed, it checks the network factors on optimum certificates and on mutated
short choices against a separate global certificate predicate. It verifies
costs, the $8K$ activity bound, and round trips through the grouped subset
encoding.

\begin{center}
\begin{tabular}{@{}lr@{}}
\toprule
Check & Count\\
\midrule
Strings with an independent optimum & 332\\
Local configurations compared with the global predicate & 5,004\\
Valid configurations among those checks & 561\\
Corrupted deterministic summary bits rejected & 553\\
Grouped subset encoding round trips & 562\\
Exhaustively enumerated selector assignments & 163,840\\
Valid certificates in the exhaustive enumeration & 60\\
Individual subset ranks and inverse ranks checked & 8,191\\
\bottomrule
\end{tabular}
\end{center}

The exhaustive selector tests use every binary string of length three and
the string $aaaa$, with threshold two. They compare every selector choice
with the independent certificate predicate and recover the oracle optimum.
They also check positive counting at every budget up to the input length
with the stated modulus bound. Subset ranking is checked for every subset
of universes of sizes zero through twelve.

Two additional cases exercise structural conditions directly. One has
matching coverage counts and available dictionary contents but an $R_1$
token crossing adjacent dictionary boundaries; the position check rejects
it. A nested example uses $A\to ab$, $B\to AAA$, and $S\to BBBB$, generating
$(ab)^{12}$ with size nine. At threshold four, its dictionary and start
rule contain long intervals while its nested definitions supply short
references. The network accepts it with the correct cost and activity bound.

These finite checks test the implementation of the reduction; they are
not part of the proof. The
general correctness and complexity statements rest on
Section~\ref{sec:sparse-interval-tensor} and its stated external theorem.
No Lean formalization of the new reduction is claimed. The script and
recorded results are in \texttt{verification/}; the original binary-code
check remains in \texttt{anc/check\_sync.py}.

\clearpage
\section*{Acknowledgements}
\paragraph{Use of generative AI.}
The author delegated the development of the binary-alphabet NP-completeness proof and of the tunable-alphabet reductions of Section~\ref{sec:eth} to OpenAI GPT-6 Pro, which produced the corresponding derivations and proofs. Generative AI tools were also used for language editing and manuscript preparation. The uniform sparse-interval tensor algorithm was developed through discussions with OpenAI GPT-5.6 Pro and GPT-6 Pro. The quantitative transfer analysis, independent finite verification, and revision integration were developed with assistance from OpenAI Codex. The author takes full responsibility for the mathematical content and the final manuscript.

\end{document}